%% file: bgch.tex
\documentclass[sigconf]{acmart}
\usepackage{multirow}
\usepackage{float}
\usepackage{subfigure}
\usepackage{enumerate}
\usepackage{enumitem}
\usepackage{stfloats}
\usepackage{hyperref}
\usepackage{breakurl}
\usepackage{diagbox}
\usepackage{bbding}
\usepackage{titlesec}
\usepackage{bm}
\usepackage{verbatimbox}
\usepackage{colortbl}
\usepackage{booktabs}
\usepackage{balance}
\usepackage{makecell}
\usepackage{titlesec}
\usepackage{amsmath,amsfonts,amsthm}
\usepackage{color}
\usepackage{graphicx}
\usepackage{caption}
\usepackage{xcolor}
\usepackage{colortbl}
\usepackage{amsthm}
\usepackage{ulem}

\usepackage[linesnumbered,ruled,vlined]{algorithm2e}
\usepackage{cleveref}

\AtBeginDocument{%
  \providecommand\BibTeX{{%
    Bib\TeX}}}

\definecolor{gray}{RGB}{221, 221, 221}

\setcopyright{acmcopyright}
\copyrightyear{2023}
\acmYear{2023}
\setcopyright{acmlicensed}\acmConference[WWW '23]{Proceedings of the ACM Web Conference 2023}{May 1--5, 2023}{Austin, TX, USA}
\acmBooktitle{Proceedings of the ACM Web Conference 2023 (WWW '23), May 1--5, 2023, Austin, TX, USA}
\acmPrice{15.00}
\acmDOI{10.1145/3543507.3583219}
\acmISBN{978-1-4503-9416-1/23/04}

\theoremstyle{stype}
\newtheorem{thm}{Theorem}

\crefname{section}{§}{§§}
\Crefname{section}{§}{§§}

\newcolumntype{L}{@{}>{\kern\tabcolsep}l<{\kern\tabcolsep}}

\newcommand{\tabincell}[2]{\begin{tabular}{@{}#1@{}}#2\end{tabular}}

\def\model{BGCH}

\def\emb{\boldsymbol} 

\DeclareMathOperator\sign{sign}
\DeclareMathOperator\argmin{argmin}
\DeclareMathOperator\diag{diag}

\SetKwComment{Comment}{$\triangleright$\ }{}

\newenvironment{sequation}{\begin{equation}\setlength\abovedisplayskip{2pt}\setlength\belowdisplayskip{2pt}}{\end{equation}}

\makeatletter
\newcommand\notsotiny{\@setfontsize\notsotiny\@vipt\@viipt}
\makeatother

\def\BibTeX{{\rm B\kern-.05em{\sc i\kern-.025em b}\kern-.08em
    T\kern-.1667em\lower.7ex\hbox{E}\kern-.125emX}}

\def\drop{\color{blue} }

\begin{document}

\title{Bipartite Graph Convolutional Hashing for Effective and Efficient Top-N Search in Hamming Space
}

\author{
Yankai Chen$^1$,
Yixiang Fang$^2$,
Yifei Zhang$^1$,
Irwin King$^1$
}

 \affiliation{
  \city{\{ykchen, yfzhang, king\}@cse.cuhk.edu.hk \quad fangyixiang@cuhk.edu.cn}\\
  \country{$^1$The Chinese University of Hong Kong \quad $^2$The Chinese University of Hong Kong, Shenzhen}
}
\renewcommand{\shortauthors}{Yankai Chen et al.}

\input{abstract}

\vspace{-2cm}



\maketitle

\input{intro}

\input{related}

\input{problem}

\input{method}

\input{exp}

\input{con}

\begin{acks}
We thank anonymous reviewers for their insightful comments and suggestions.
Yankai Chen, Yifei Zhang and Irwin King were supported by the National Key Research and Development Program of China (No. 2018AAA0100204) and by the Research Grants Council of the Hong Kong Special Administrative Region, China (CUHK 2410021, Research Impact Fund, No. R5034-18).
Yixiang Fang was supported by NSFC Grant (62102341).
\end{acks}

\clearpage

\bibliographystyle{ACM-Reference-Format}
\bibliography{ref}

\clearpage
\input{app}

\end{document}

%% file: abstract.tex
\begin{abstract}
Searching on bipartite graphs is basal and versatile to many real-world Web applications, e.g., online recommendation, database retrieval, and query-document searching.
Given a query node, the conventional approaches rely on the similarity matching with the vectorized node embeddings in the continuous Euclidean space. 
To efficiently manage intensive similarity computation, developing hashing techniques for graph-structured data has recently become an emerging research direction.
Despite the retrieval efficiency in Hamming space, prior work is however confronted with \textit{catastrophic performance decay}.  
In this work, we investigate the problem of hashing with Graph Convolutional Network on bipartite graphs for effective Top-N search.
We propose an end-to-end \textit{\underline{B}ipartite \underline{G}raph \underline{C}onvolutional \underline{H}ashing} approach, namely \model, which consists of three novel and effective modules: (1) \textit{adaptive graph convolutional hashing}, (2) \textit{latent feature dispersion}, and (3) \textit{Fourier serialized gradient estimation}.
Specifically, the former two modules achieve the substantial retention of the structural information against the inevitable information loss in hash encoding;
the last module develops Fourier Series decomposition to the hashing function in the frequency domain mainly for more accurate gradient estimation.
The extensive experiments on six real-world datasets not only show the performance superiority over the competing hashing-based counterparts, but also demonstrate the effectiveness of all proposed model components contained therein.
\end{abstract}

%% file: intro.tex
\section{\textbf{Introduction}}

Bipartite graphs are ubiquitous in the real world for the ease of modeling various Web applications, e.g., as shown in Figure~\ref{fig:intro}(a), user-product recommendation~\cite{ma2020probabilistic,zhang2019star} and online query-document matching~\cite{zhang2019doc2hash}.
A fundamental task, \textit{Top-N search}, is to filter out N best-matched graph nodes for a query node, e.g., recommending Top-N attractive products to a target user in the user-product graph.
With the development of the recent machine learning research, learning vectorized representations (\textit{a.k.a.} embeddings) nowadays has become one of the standard procedures for similarity matching~\cite{grover2016node2vec,cheng2018aspect}.
Among existing techniques, graph-based neural methods, i.e., \textit{Graph Convolutional Networks} (GCNs), have recently present remarkable model performance~\cite{graphsage,lightgcn}.
Due to the ability to capture high-order connection information, GCN models can thus produce semantic enrichment to the node embeddings.
Based on the learned embeddings, similarity estimation is then exhaustively proceeded in the continuous Euclidean space.

Apart from embedding informativeness, \textit{computation latency} and \textit{embedding memory overhead} are two important criteria for realistic application deployment.
With the explosive data growth, \textit{learning to hash}~\cite{wang2017survey,jegou2010product} recently provides an alternative option to graph-based models for optimizing the model scalability.
Generally, it learns to convert the vectorized list of continuous values into the finite binarized hash codes.
In lieu of using \textit{full-precision}\footnote{\scriptsize The term ``full-precision'' generally refers to single-precision and double-precision. And we use float32 by default throughout this work for illustration.} 
embeddings, the learned hash codes have the promising potential to achieve, not only the space reduction, but also the computation acceleration for Top-N object matching and retrieval in the Hamming space.

\begin{figure}[tp]
\begin{minipage}{0.5\textwidth}
\includegraphics[width=3.4in]{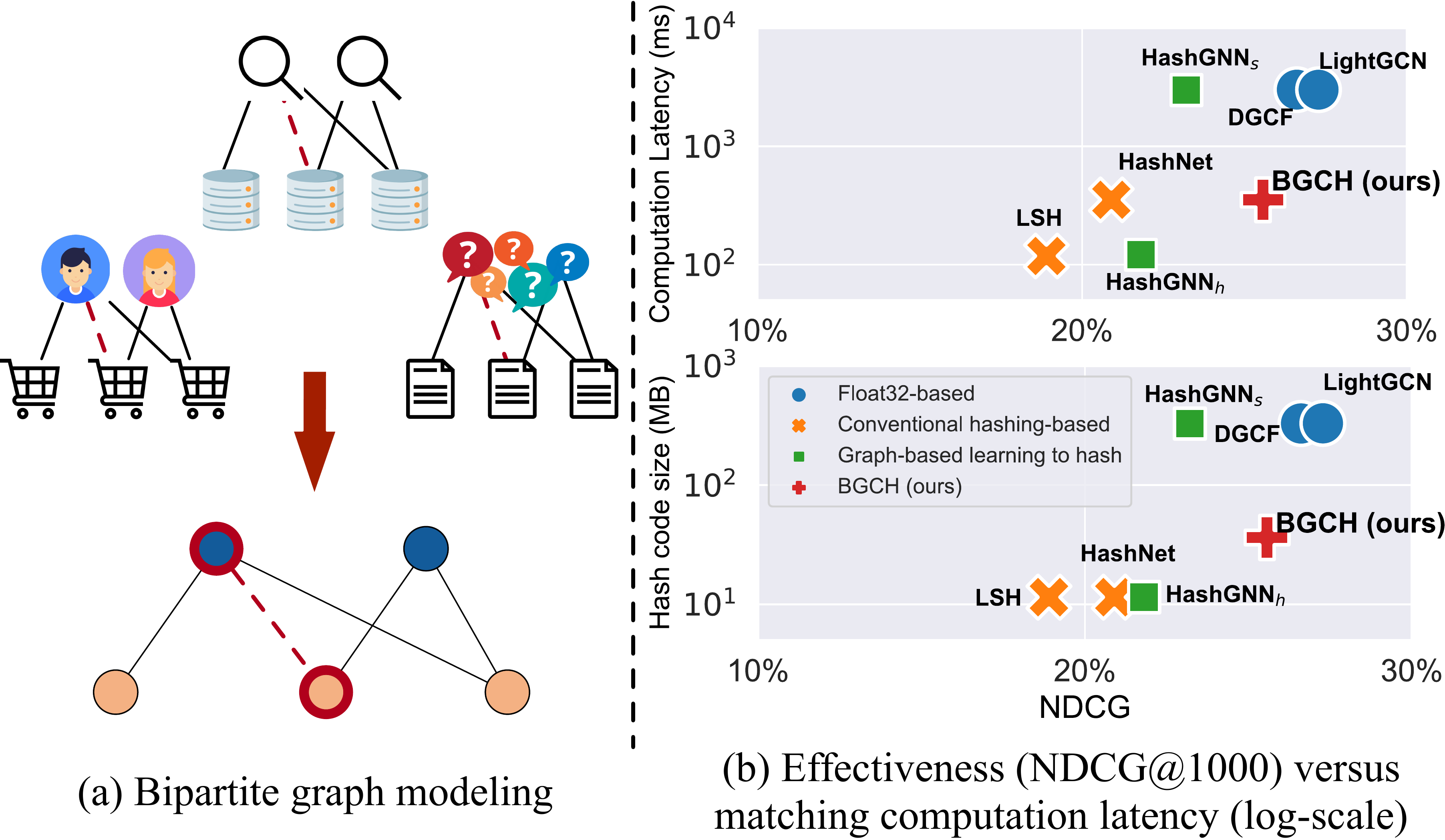}
\end{minipage} 
\vspace{-0.1in}
\caption{Illustration of bipartite graph modeling and overall model performance visualization on Dianping dataset.}
\label{fig:intro}
\end{figure}

Despite the promising advantages of bridging GCNs and learning to hash,
simply stacking these two techniques is trivial and thus falls short of performance satisfaction with several inadequacies:
\begin{itemize}[leftmargin=*]

\item \textbf{Coarse-grained similarity measurement.}
Compared to continuous embeddings, hash codes with the same vector dimension are naturally \textit{less expressive} with finite encoding permutation in Hamming space (e.g., $2^d$ if the dimension is $d$).
Consequently, this leads to a coarse-grained estimation of the pairwise node similarity, thus drawing a conspicuous performance decay with inaccurate Top-N matching.

\item \textbf{Feature erosion issue.}
Recent work~\cite{qin2020forward,rastegari2016xnor,lin2017towards,hashgnn} usually adopts $\sign(\cdot)$ function for $O(1)$ complexity encoding.
However, hashing via $\sign(\cdot)$ will inevitably smooth the embedding feature informativeness, via converting each digit of continuous embeddings into the hamming space, no matter what specific value it used to be.
Thus the latent features in these learned hash embeddings become less informative, and existing models lack certain mechanisms to hedge the feature erosion in hashing.

\item \textbf{Intractable model optimization.}
Since $\sign(\cdot)$ is not differentiable at 0 and its derivatives are 0 anywhere else, previous models usually use \textit{visually similar} but not necessarily \textit{theoretically relevant} functions, e.g., $\tanh(\cdot)$, for gradient estimation.
This may lead to inconsistent optimization directions in model training.
Moreover, because of the embedding discreteness, the associated loss landscape\footnote{\scriptsize Details about the visualization construction are attached in Appendix A.} (Figure~\ref{fig:model}(a)) are steep and bumping~\cite{bai2020binarybert}, which further increases the difficulty in optimization.
\end{itemize}

In this paper, we study the problem of learning to hash with Graph Convolutional Network (GCN) on bipartite graphs for effective Top-N Hamming space search.
We propose a model namely \textit{\underline{B}ipartite \underline{G}raph \underline{C}onvolutional \underline{H}ashing} (BGCH), with three effective modules: (1) \textit{adaptive graph convolutional hashing}, (2) \textit{latent feature dispersion}, and (3) \textit{Fourier serialized gradient estimation}.
While the former two modules significantly enrich the informativeness and expressivity to the learned hash codes, the last one provides an accordant and tractable optimization flow in forward and backward propagation of model optimization.
Concretely:
\begin{itemize}[leftmargin=*]
\item \textbf{Adaptive graph convolutional hashing.}
Our first module designs a \textit{topology-aware} convolutional hashing that employs the layer-wise hash encoding (from low- to high-order sub-structures of bipartite graphs) to consecutively binarize the node features with different semantics.
To boost the expressivity, the convolutional hashing is equipped an effective approximation technique for \textit{embedding rescaling}, which does not undermine the efficiency of Hamming distance computation.
Intuitively, these two designs make the learned hash codes more informative and expressive for preserving fine-grained similarity in the Hamming space.

\item \textbf{Latent feature dispersion}.
Our second module, i.e., feature dispersion, aims to hedge the inevitable information loss from the numerical binarization.
In conventional continuous embeddings, major features however condense in a small region of embedding structures.
Since these vectorized latent features tend to be inevitably smoothed by the hashing discreteness, it is natural to preserve information as much as possible by spreading out those decisive features to the majority of embedding dimensions, instead of the very few of them.
To achieve this, our proposed module aims to explicitly disperse informative latent node features, which can be further diffused to each convolutional layer when exploring the bipartite graph.

\item \textbf{Fourier serialized gradient estimation.}
Furthermore, to provide accurate gradient estimation, \model~proposes to decompose $\sign(\cdot)$ function with Fourier Series in the frequency domain.
Compared to existing gradient estimators~\cite{qin2020forward,gong2019differentiable,darabi2018bnn,sigmoid,RBCN}, this estimator better follows the \textit{main direction} of factual gradients to enable an accordant and tractable model optimization in forward and backward propagation.
With the limited number of decomposition terms, \model~can well provide more accurate gradient estimation to $\sign(\cdot)$ within the acceptable training cost.
\end{itemize}

Based on the learned hash codes, \model~maintains moderate resource consumption whilst providing substantial performance improvement in Top-N Hamming space retrieval.
The quality-cost trade-off is summarized in Figure~\ref{fig:intro}(b), which compares \model~ against a list of representative counterparts (\textit{including float32-based and hashing-based}) on a real-world bipartite graph with over 10 million observed edges (experimental details are reported in~\cref{sec:exp_setup}).
As the lower-right corner of Figure~\ref{fig:intro}(b) indicates the ideal optimal performance, \model~can deliver over 8$\times$ computation acceleration and space reduction relative to existing full-precision models, while being more effective than each hash-based method (\cref{sec:exp_topn} and \cref{sec:exp_full}).
To summarize, our main contributions are organized as follows:
\begin{itemize}[leftmargin=*]
\item 
We study the problem of learning to hash with Graph Convolutional Network on bipartite graphs.
We propose a novel approach \model~with three effective modules for effective and efficient Top-N search in Hamming space (\cref{sec:method}).

\item We conduct extensive experiments on six real-world datasets to evaluate the retrieval quality.
In-depth analyses are also provided towards the necessity of all proposed model components from both technical and empirical perspectives (\cref{sec:exp}).


\item We theoretically prove the model effectiveness and provide complexity analyses in terms of time and space costs (Appendix C).
\end{itemize}

%% file: related.tex
\section{\textbf{Related Work}}
\label{sec:work}

{\textbf{Graph convolution network (GCN)}.}
Early work studies the graph convolutions mainly on the \textit{spectral domain}, such as Laplacian eigen-decomposition~\cite{bruna2013spectral} and Chebyshev polynomials~\cite{defferrard2016convolutional}.
One major issue is that these models are usually computationally expensive. 
To tackle this problem, \textit{spatial-based} GCN models are proposed to re-define the graph convolution operations by aggregating the embeddings of neighbors to refine and update the target node's embedding.
Due to its scalability to large graphs, spatial-based GCN models are widely used in various applications~\cite{lightgcn,ngcf,graphsage}. 
For example, to capture high-order structural information, NGCF~\cite{ngcf} and LightGCN~\cite{lightgcn} learn the collaborative filtering signals on bipartite interaction graphs for recommendation.
Despite the effectiveness in embedding latent features for graph nodes, they usually suffer from inference inefficiency due to the high computational cost of similarity calculation between continuous embeddings~\cite{hashgnn}.
To address this issue, \textit{learning to hash} provides the feasibility.

{\textbf{Learning to hash}.}
Learning to hash models are promising to achieve computation acceleration and storage reduction for general information retrieval and processing tasks~\cite{hu2021semi,hu2020selfore,li2020unsupervised,lightgcn,gao2020discern,chen2022repo4qa,ma2019hierarchical}.
More than reducing conflicts~\cite{kraska2018case},  similarity-preserving hashing maps high-dimensional dense vectors to a low-dimensional Hamming space for efficiently processing downstream tasks.
A representative model is Locality Sensitive Hashing (LSH)~\cite{lsh} that uses random projections as the hash functions.
Recent work focuses on integrating the deep neural network architectures for model improvement~\cite{wang2017survey}.
They inspire a series of follow-up work for various tasks, such as fast retrieval of images~\cite{qin2020forward,lin2017towards,hashnet}, documents~\cite{li2014two,chen2022effective}, categorical information~\cite{kang2021learning}, e-commerce products~\cite{zhang2017discrete,chen2022learning}.

To leverage hashing techniques with GCNs, the state-of-the-art work HashGNN~\cite{hashgnn} investigates \text{learning to hash} for online matching and recommendation.
Specifically, HashGNN consectively combines the GraphSage~\cite{graphsage} as the embedding encoder and learning to hash method to get the corresponding binary encodings afterwards.
Its hash encoding process only proceeds at the end of multi-layer graph convolutions, i.e., using the aggregated output of GraphSage for representation binarization. 
However, this fails to capture intermediate semantics from nodes' different layers of receptive fields~\cite{kipf2016semi}.
The other issue of HashGNN is using \textit{Straight-Through Estimator (STE)}~\cite{bengio2013estimating} to assume all gradients of $\sign(\cdot)$ as 1 in backpropagation.
However, the integral of 1 is a certain linear function other than the $\sign(\cdot)$, whereas this may lead to inconsistent optimization directions in the model training.
To address these issues, our model \model~ is proposed with effectiveness justification in~\cref{sec:exp}.

%% file: problem.tex
\section{Preliminaries and Problem Formulation}
\label{sec:pre}

\noindent{\textbf{Graph Convolution Network (GCN).}}
The general idea of GCN is to learn node embeddings by \textit{iteratively propagating and aggregating} latent features of node neighbors via the graph topology~\cite{wu2019simplifying,lightgcn,kipf2016semi}:
\begin{equation}
\setlength\abovedisplayskip{2pt}
\setlength\belowdisplayskip{2pt}
\boldsymbol{V}_x^{(l)} = AGG\left(\boldsymbol{V}_x^{(l-1)}, \{\boldsymbol{V}_z^{(l-1)}: z \in \mathcal{N}(x)\}\right),
\end{equation}%
where {\small$\emb{V}_x^{(l)} \in \mathbb{R}^d$} denotes node $x$'s embedding after $l$-th iteration of graph convolutions, indexed in the embedding matrix {\small $\emb{V}$}. 
{\small $\mathcal{N}(x)$} is the set of $x$'s neighbors.
Function $AGG(\cdot, \cdot)$ is the information aggregation function, with several implementations in previous work~\cite{kipf2016semi,graphsage,gat,xu2018powerful}, mainly aiming to transform the center node feature and the neighbor features.
In this work, we adopt the graph convolution paradigm from the state-of-the-art model LightGCN~\cite{lightgcn}.

\vspace{0.05in}

\noindent\textbf{Bipartite Graph and Adjacency Matrix.} 
The bipartite graph is denoted as $\mathcal{G}=\{\mathcal{V}_1, \mathcal{V}_2, \mathcal{E}\}$, where $\mathcal{V}_1$ and  $\mathcal{V}_2$ are two \textit{disjoint} node sets and $\mathcal{E}$ is the set of edges between nodes in $\mathcal{V}_1$ and $\mathcal{V}_2$.
We can use $\emb{Y} \in \mathbb{R}^{|\mathcal{V}_1|\times |\mathcal{V}_2}|$ to indicate the edge transactions, where 1-valued entries, i.e., $\emb{Y}_{x,y}=1$, indicate there is an observed edge between nodes $x\in \mathcal{V}_1$ and $y \in \mathcal{V}_2$, otherwise $\emb{Y}_{x,y}=0$.
Then the adjacency matrix $\emb{A}$ of the whole graph can be defined as:
\begin{equation}
\setlength\abovedisplayskip{2pt}
\setlength\belowdisplayskip{2pt}
\emb{A} = 
\begin{bmatrix}
\emb{0} & \emb{Y} \\
\emb{Y}^T & \emb{0}
\end{bmatrix}.
\end{equation}%


\noindent\textbf{Problem Formulation.}
Give a bipartite graph $\mathcal{G}=\{\mathcal{V}_1, \mathcal{V}_2, \mathcal{E}\}$ and its adjacency matrix $\emb{A}$, we devote to learn a hashing function:
\begin{equation}
\setlength\abovedisplayskip{2pt}
\setlength\belowdisplayskip{2pt}
F(\emb{A} | \Theta) \rightarrow \mathcal{Q},
\end{equation}
where $\Theta$ is the set of all learnable parameters. 
Given two nodes in the bipartite graph, e.g., $x\in \mathcal{V}_1$ and $y \in \mathcal{V}_2$, their hash codes are $\mathcal{Q}_x$ and $\mathcal{Q}_y$.
Then the probability of edge existence $\widehat{\emb{Y}}_{x,y}$ between nodes $x\in \mathcal{V}_1$ and $y \in \mathcal{V}_2$ can be effectively and efficiently measured by the hash codes $\mathcal{Q}_x$ and  $\mathcal{Q}_y$, i.e., $\widehat{\emb{Y}}_{x,y}$ = $f(\mathcal{Q}_x, \mathcal{Q}_y)$ where $f$ is a score function.
Intuitively, the larger value $\widehat{\emb{Y}}_{x,y}$ is, the more likely $x$ and $y$ are matched, i.e., an edge between $x$ and $y$ exists.
Explanations of key notations used in this paper are attached in Appendix B.

%% file: method.tex
\section{\model: Methodology}
\label{sec:method}
\begin{figure*}[tp]
\hspace{-0.1in}
\begin{minipage}{1\textwidth}
\includegraphics[width=7.1in]{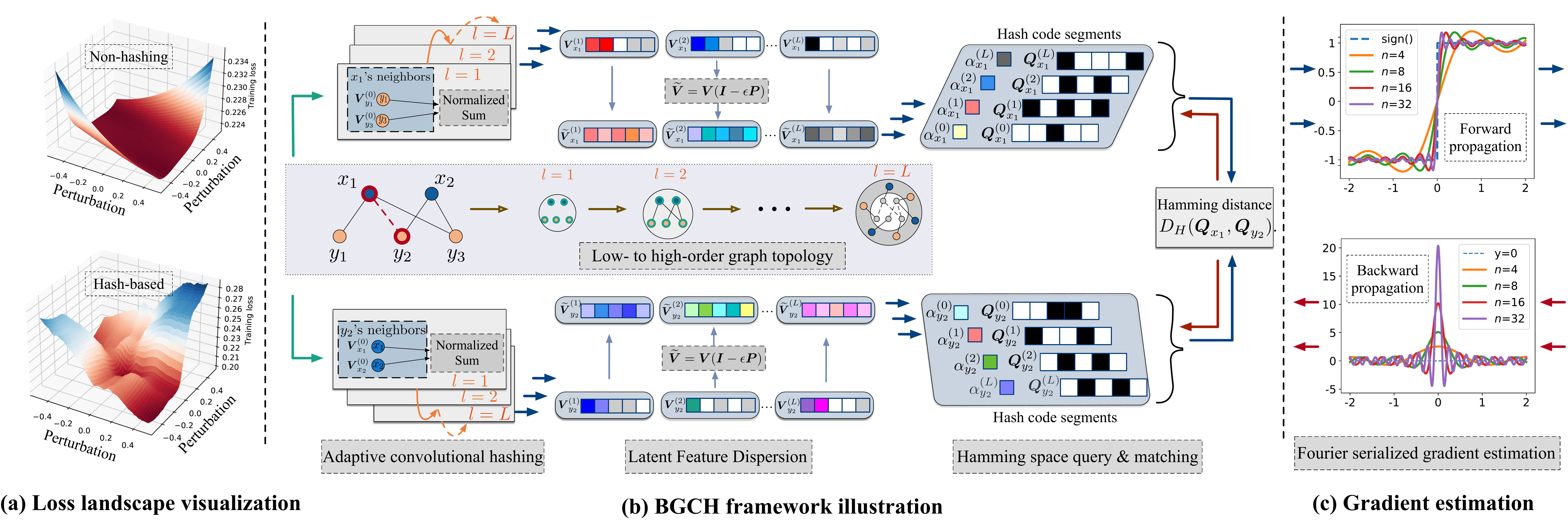}
\end{minipage} 
\vspace{-0.1in}
\caption{(a) Visualized loss landscape comparison; (b) \model~model framework (best view in color); (c) Fourier Serialized gradient estimation in forward and bachward propagation. }
\label{fig:model}
\end{figure*}

\subsection{Overview}
We formally introduce our \model~ model.
Notice that since the effect of feature dispersion module propagates along with convolutional hashing, we then introduce these modules in the following order:
(1) \textit{latent feature dispersion} (\cref{sec:fd}) aims to disperse the embedded features into wider embedding structures to hedge the inevitable information loss in hashing;
(2) \textit{adaptive graph convolutional hashing} (\cref{sec:hashing}) provides an effective encoding approach to significantly improve the hashed feature expressivity whilst maintaining the matching efficiency in the hamming space;
(3) \textit{Fourier serialized gradient estimation} (\cref{sec:ge}) introduces the Fourier Series decomposition for $\sign(\cdot)$ in the frequency domain to provide more accurate gradient approximation.
Based on the learned hash codes, \model~ develops efficient online matching with the Hamming distance measurement (\cref{sec:score}).
Our model illustration is attached in Figure~\ref{fig:model}(b).


\input{method_fd}

\input{method_hashing}

\input{method_gradient}

\input{method_optimize}

%% file: method_fd.tex
\subsection{Latent Feature Dispersion}
\label{sec:fd}

To tackle the feature erosion issue, we seek to disperse the embedded features as one effective strategy to hedge the inevitable information loss caused by numerical binarization.
From the perspective of singular value decomposition (SVD), singular values and corresponding singular vectors reconstruct the original matrix;
normally, large singular values can be interpreted to associate with \textit{major feature structures} of the matrix~\cite{wei2018grassmann}. 
Since we want to avoid condensing and gathering informative features in (relatively small) embedding sub-structures, it is natural to bridge the target by working on these singular values.
Hence, based on this intuition, we aim to \textit{normalize singular values for equalizing their respective contributions in constituting latent features}.
To achieve this, Power Normalization~\cite{koniusz2016higher,zhang2022spectral} is one of the solutions that tackle related problems such as feature imbalance~\cite{koniusz2018deeper}.
Inspired by the recent approximation attempt~\cite{yu2020toward}, we now introduce a lightweight feature dispersion technique in graph convolution as follows.

Concretely, let $\emb{I}$ denote the identity matrix, we start from generating a \textit{standard normal random vector} $\emb{p}^{(0)}$$\sim$$\mathcal{N}(\emb{0}, \emb{I})$ where $\emb{p}^{(0)}$ $\in$ $\mathbb{R}^{c}$.
Based on the embedding matrix to conduct feature dispersion, e.g., let $\emb{V}= \emb{V}^{(0)}$, we compute the desired \textbf{dispersing vector} $\emb{p}^{(k)}$ by iteratively performing $\emb{p}^{(k)} = \emb{V}^\mathsf{T}\emb{V}\emb{p}^{(k-1)}$.
The iteration for generating dispersing vectors is independent of the graph convolution iterations\footnote{\scriptsize In our work, we set $K \leq L$ mainly to enable the associated complexity of dispersing vector generation is upper bounded by the graph convolution complexity.}.
We have the projection matrix $\emb{P}$ of $\emb{p}^{(K)}$ via:
\begin{sequation}
\label{eq:projection}
\emb{P} = \frac{\emb{p}^{(K)}\emb{p}^{(K)^\mathsf{T}}}{||\emb{p}^{(K)}||_2^2}.
\end{sequation}%
Then we have the feature-dispersed representation matrix with the hyper-parameter $\epsilon$ $\in$ $(0,1)$ as follows:
\begin{sequation}
\label{eq:disperse}
\widetilde{\emb{V}} = \emb{V}(\emb{I} - \epsilon \emb{P}).
\end{sequation}%
Consequently, integrating the dispersed matrix $\widetilde{\emb{V}}$, we have the \textbf{feature-dispersed} graph convolution as:
\begin{sequation}
\label{eq:fdconv}
\widetilde{\emb{V}}^{(l+1)} = (\emb{D}^{-\frac{1}{2}} \emb{A} \emb{D}^{-\frac{1}{2}} )\widetilde{\emb{V}}^{(l)},  \text{ where } \widetilde{\emb{V}}^{(0)} = \emb{V}^{(0)}(\emb{I} - \epsilon \emb{P}).
\end{sequation}%
Note that we explicitly conduct this feature dispersion operation one time only at the initial step, i.e., {\small $\widetilde{\emb{V}}^{(0)}$}, and, more importantly, such feature dispersion can be diffused via the multi-layer graph convolutions from $0$ to $L$.
Compare to the unprocessed embedding counterpart, e.g., {\small${\emb{V}}^{(l)}$}, embedding matrix {\small$\widetilde{\emb{V}}^{(l)}$} at each layer presents a dispersed feature structure with a \textit{more balanced distribution of singular values in expection}. 
We formally explain this as follows: 
\vspace{-0.05in}
\begin{thm}[\textbf{Feature Dispersion}]
\label{tm:svd}
Let ${\emb{V}}^{(l)} = \emb{U}_1\emb{\Sigma}\emb{U}_2^\mathsf{T}$, where $\emb{U}_1$ and $\emb{U}_2$ are unitary matrices and descending singular value matrix $\emb{\Sigma} = \diag(\sigma_1, \sigma_2, \cdots, \sigma_c)$.  
Then $\mathbb{E}({\small\widetilde{\emb{V}}^{(l)}}) = \emb{U}_1\emb{\Sigma}\emb{\Sigma}_{\mu}\emb{U}_2^\mathsf{T}$ where $\emb{\Sigma}_{\mu} = \diag(\mu_1, \mu_2, \cdots, \mu_c)_{0<\mu_{1 \cdots c}<1}$ is in ascending order.
\end{thm}

\vspace{-0.05in}
Intuitively, given the same orthonormal bases, compared to {\footnotesize$\emb{V}^{(l)}$}, it is harder in expectation to reconstruct {\footnotesize$\widetilde{\emb{V}}^{(l)}$} with informative features being dispersed out in larger matrix sub-structures.
This eventually provides the functionality to hedge the information loss in numerical binarization. 
We attach the theorem proof in Appendix C and evaluate the module effectiveness later in~\cref{sec:ablation}.

%% file: method_hashing.tex
\subsection{Adaptive Graph Convolutional Hashing}
\label{sec:hashing}
One feasible solution for increasing expressivity and smoothing loss landscapes is to include the \textit{relaxation strategy}.
Hence, apart from the topology-aware embedding binarization with $\sign(\cdot)$:
\begin{equation}
\setlength\abovedisplayskip{2pt}
\setlength\belowdisplayskip{2pt}
\label{eq:hashing}
\emb{Q}_{x}^{(l)} = \sign(\widetilde{\emb{V}}_{x}^{(l)}),
\end{equation}
our model \model~additionally computes a layer-wise positive rescaling factor for each node, e.g., $\alpha_x^{(l)} \in \mathbb{R}^+$, such that $\widetilde{\emb{V}}^{(l)}_x \approx$ $\alpha_x^{(l)} \emb{Q}^{(l)}_x$.
In this work, we introduce a simple but effective approach to directly calculate the rescaling factors as follows: 
\begin{equation}
\setlength\abovedisplayskip{2pt}
\setlength\belowdisplayskip{2pt}
\label{eq:rescale}
\alpha_x^{(l)} = \frac{1}{d} ||\widetilde{\emb{{V}}}_x^{(l)}||_1.
\end{equation}
Instead of setting these factors as learnable, such deterministic computation substantially prunes the parameter search space while attaining the adaptive approximation functionality for different graph nodes. 
We demonstrate this in~\cref{sec:ablation} of experiments.

After $L$ iterations of feature propagation and hashing, we obtain the table of \textbf{adaptive hash codes} $\mathcal{Q} = \{\emb{\alpha}, \emb{Q}\}$, where $\emb{\alpha} \in \mathbb{R}^{(|\mathcal{V}_1|+|\mathcal{V}_2|)\times (L+1)}$ and $\emb{Q} \in \mathbb{R}^{(|\mathcal{V}_1|+|\mathcal{V}_2|)\times d}\}$.
For each node $x$, its corresponding hash codes are indexed and assembled:
\begin{equation}
\setlength\abovedisplayskip{2pt}
\setlength\belowdisplayskip{2pt}
\emb{\alpha}_x = \alpha_x^{(0)} || \alpha_x^{(1)} || \cdots || \alpha_x^{(L)}, \text{  and  } \emb{{Q}}_x = \emb{Q}_x^{(0)} || \emb{Q}_x^{(1)} || \cdots || \emb{Q}_x^{(L)}.
\end{equation}
Intuitively, the hash code table $\mathcal{Q}$ represents the bipartite structural information that is propagated back and forth at different iteration steps $l$, i.e., from $0$ to the maximum step $L$.
It not only tracks the intermediate knowledge hashed for all graph nodes, but also maintains the value approximation to their original continuous embeddings, e.g., {\footnotesize $\widetilde{\emb{{V}}}_x^{(l)}$}.
In addition, with the slightly more space cost (complexity analysis in Appendix C, such detached hash encoding approach still supports the bitwise operations (~\cref{sec:score}) for accelerating inference and matching.

%% file: method_gradient.tex
\subsection{Fourier Serialized Gradient Estimation}
\label{sec:ge}

To provide the accordant gradient estimation for hash function $\sign(\cdot)$, we approximate it by introducing its Fourier Series decomposition in the frequency domain. 
Specifically, $\sign(\cdot)$ can be viewed as a special case of the periodical Square Wave Function $t(x)$ within the length $2H$, i.e., $\sign(\phi) = t(\phi)$, $|\phi| < H$.  
Since $t(x)$ can be decomposed in Fourier Series, we shall have: 
\begin{equation}
\setlength\abovedisplayskip{2pt}
\setlength\belowdisplayskip{2pt}
\sign(\phi) = \frac{4}{\pi}\sum_{i=1,3,5,\cdots}^{+\infty}\frac{1}{i}\sin(\frac{\pi i\phi}{H}), {\rm \ \ where \ \ } |\phi| < H.
\end{equation}


Fourier Series decomposition of $\sign(\cdot)$ with infinite terms is a lossless transformation~\cite{rust2013convergence}.
Thus, as shown in Figure~\ref{fig:model}(c), we can set the finite expanding term $n$ to obtain its approximation version as follows: 
\begin{equation}
\setlength\abovedisplayskip{2pt}
\setlength\belowdisplayskip{2pt}
{\sign(\phi)} \doteq \frac{4}{\pi}\sum_{i=1,3,5,\cdots}^{n}\frac{1}{i}\sin(\frac{\pi i\phi}{H}).  \\
\end{equation}
The corresponding derivatives can be derived accordingly as:
\begin{equation}
\setlength\abovedisplayskip{2pt}
\setlength\belowdisplayskip{2pt}
\label{eq:gradient}
\frac{\partial{{\sign(\phi)}}}{\partial \phi}   \doteq \frac{4}{H} \sum_{i=1,3,5,\cdots}^{n} \cos(\frac{\pi i\phi}{H}). 
\end{equation}

Different from other gradient estimators such as tanh-alike~\cite{gong2019differentiable,qin2020forward} and SignSwish~\cite{darabi2018bnn}, approximating $\sign(\cdot)$ function with its Fourier Series will not corrupt the main direction of factual gradients in model optimization~\cite{xu2021learning}.
This is beneficial to bridge a coordinated transformation from the continuous values to its corresponding binarization for node representations, which significantly retains the discriminability of binarized representations and produces better retrieval accuracy accordingly.
We present this performance comparison in~\cref{sec:fs_exp} of experiments. 
To summarize, as shown in Equation~(\ref{eq:formal_grad}), to learn and optimize the binarized embeddings for graph nodes, we apply the strict $\sign(\cdot)$ function for forward propagation and estimate the gradients $\frac{\partial\sign(\phi)}{\partial \phi}$ for backward propagation.
\begin{equation}
\setlength\abovedisplayskip{2pt}
\setlength\belowdisplayskip{2pt}
\label{eq:formal_grad}
\left\{ 
\begin{aligned}
& \boldsymbol{Q}^{(l)} = \sign(\phi),  &\text{Forward propagation.} \\
& \frac{\partial \boldsymbol{Q}^{(l)}}{\partial \phi} \doteq \frac{4}{H} \sum_{i=1,3,5,\cdots}^{n} \cos(\frac{\pi i\phi}{H}). & \text{Backward propagation.}
\end{aligned}
\right.
\end{equation}

%% file: method_optimize.tex
\subsection{Score Prediction and Model Optimization}
\label{sec:score}

\subsubsection{\textbf{Matching score prediction.}}
\label{sec:score_computation}
Given two nodes $x \in \mathcal{V}_1$ and $y \in \mathcal{V}_2$, one natural manner to implement the score function is \textit{inner-product}, mainly for its simplicity as:
\begin{equation}
\setlength\abovedisplayskip{2pt}
\setlength\belowdisplayskip{2pt}
\label{eq:inner_score}
\widehat{\emb{Y}}_{x,y} =  (\alpha_x\emb{Q}_x)^\mathsf{T} \cdot (\alpha_y\emb{Q}_y).
\end{equation}
However, the inner product in Equation~(\ref{eq:inner_score}) is still conducted in the (continuous) Euclidean space with \textit{full-precision arithmetics}.
To bridge the connection between the inner product and Hamming distance measurement, we introduce Theorem~\ref{tm:equal} as follows:

\begin{thm}[\textbf{Hamming Distance Matching}]
\label{tm:equal}
Given two hash codes, we have $(\alpha_x\emb{Q}_x)^\mathsf{T} \cdot (\alpha_y\emb{Q}_y)$ $=$ $\alpha_x\alpha_y$ $(d - 2D_{H}(\emb{Q}_x, \emb{Q}_y))$.
\end{thm}

$D_H(\cdot, \cdot)$ denotes the Hamming distance between two inputs.
Based on Theorem~\ref{tm:equal}, we transform the score computation to the Hamming distance matching. 
By doing so, we can reduce most number of the floating-point operations (\#FLOPs) in the original score computation formulation (Equation~(\ref{eq:inner_score})) to efficient hamming distance matching.
This can develop substantial computation acceleration that is analyzed in Appendix C.

\subsubsection{\textbf{Multi-loss Objective Function.}}
Our objective function consists of two components, i.e., graph reconstruction loss $\mathcal{L}_{rec}$ and BPR loss $\mathcal{L}_{bpr}$. 
Generally, these two loss functions harness the regularization effect to each other.
The intuition of such design is: 
\begin{itemize}[leftmargin=*]
\item $\mathcal{L}_{rec}$ reconstructs the observed bipartite graph topology;
\item $\mathcal{L}_{bpr}$ ranks the matching scores computed from the hash codes. 
\end{itemize}
Concretely, we implement $\mathcal{L}_{rec}$ with Cross-entropy loss:  
\begin{equation}
\label{eq:rec}
\setlength\abovedisplayskip{2pt}
\setlength\belowdisplayskip{2pt}
\resizebox{1\linewidth}{!}{$
\displaystyle
\mathcal{L}_{rec} = \sum_{x \in \mathcal{V}_1} \Big(\sum_{y\in \mathcal{N}(x)} \ln\sigma\Big(({\emb{V}}^{(0)}_x)^\mathsf{T} \cdot {\emb{V}}^{(0)}_y\Big) + \sum_{y' \notin \mathcal{N}(x)}\ln\Big(1-\sigma\big(({\emb{V}}^{(0)}_x)^\mathsf{T} \cdot {\emb{V}}^{(0)}_{y'}\big)\Big)\Big),
$}
\end{equation}
where $\sigma$ is the activation function, e.g., Sigmoid.
$\mathcal{L}_{rec}$ bases on the initial continuous embeddings before the graph convolution, e.g., {\small${\emb{V}}^{(0)}_x$}, providing the most fundamental information for topology reconstruction.
As for $\mathcal{L}_{bpr}$, we employ \textit{Bayesian Personalized Ranking} (BPR) loss as:
\begin{equation}
\setlength\abovedisplayskip{2pt}
\setlength\belowdisplayskip{2pt}
\label{eq:hd-bpr}
\mathcal{L}_{bpr} = -\sum_{x \in \mathcal{V}_1} \sum_{\tiny y\in \mathcal{N}(x) \atop y'\notin \mathcal{N}(x)} \ln \sigma(\widehat{\emb{Y}}_{x,y} - \widehat{\emb{Y}}_{x,y'}).
\end{equation}
$\mathcal{L}_{bpr}$ encourages the predicted score of an observed edge to be higher than its unobserved counterparts~\cite{lightgcn}.
Let $\Theta$ denote the set of trainable embeddings regularized by the parameter $\lambda_2$ to avoid over-fitting.
our final objective function is finally defined as:
\begin{equation}
\label{eq:L}
\setlength\abovedisplayskip{2pt}
\setlength\belowdisplayskip{2pt}
\mathcal{L} = \mathcal{L}_{rec} + \lambda_1\mathcal{L}_{bpr} + \lambda_2 ||\Theta||_2^2.
\end{equation} 


So far, we have introduced all technical parts of \model~and attached the pseudocodes in Appendix B.
We present all the theorem proofs and complexity analyses in Appendix C.

%% file: exp.tex
\section{\textbf{Experimental Evaluation}}
\label{sec:exp}
We evaluate \model~to answer the following research questions:
\begin{itemize}[leftmargin=*]
\item \textbf{RQ1.} How does \model~perform compared to state-of-the-art hashing-based models in the Top-N Hamming space retrieval?

\item \textbf{RQ2.} what is the performance gap between \model~ and the full-precision models in terms of long-list retrieval quality?

\item \textbf{RQ3.} What are the benefits of proposed components in \model?

\item \textbf{RQ4.} what is the practical \model~resource consumption?

\item \textbf{RQ5.} How does the Fourier Series decomposition perform \textit{w.r.t.} retrieval accuracy and training efficiency?

\end{itemize}

\input{exp_setup}

\input{exp_topn}

\input{exp_full}

\input{exp_ablation}

\input{exp_resource}

\input{exp_fs}

%% file: exp_setup.tex
\subsection{\textbf{Experiment Setup}}
\label{sec:exp_setup}

\textbf{Datasets and evaluation metrics.} 
We include six real-world bipartite graphs in Table~\ref{tab:datasets} that are widely evaluated~\cite{lightgcn,chen2021hyper,chen2021attentive,yang2022hrcf,ngcf,zhang2022knowledge}.
We adopt evaluation protocols Recall@N and NDCG@N to measure the Top-N Hamming space ranking capability.
Dataset details and evaluation procedure are explained in Appendix D.

\begin{table}[t]
\centering
\small
\caption{The statistics of datasets.}
\vspace{-0.15in}
\label{tab:datasets}
\setlength{\tabcolsep}{0.8mm}{
\begin{tabular}{c | c | c | c | c | c | c}
\toprule 
             & {\footnotesize MovieLens}  & {\footnotesize Gowalla}   & {\footnotesize Pinterest}  &  {\footnotesize Yelp2018} & {\footnotesize AMZ-Book} & {\footnotesize Dianping}\\
\midrule[0.1pt]
    {\footnotesize |$\mathcal{V}_1$| }  & {6,040}   & {29,858}   & {55,186}   & {31,668}  &{52,643}  &{332,295}  \\ 
    {\footnotesize |$\mathcal{V}_2$| }  & {3,952}   & {40,981}   & {9,916}    & {38,048}  &{91,599}  &{1,362}  \\
\midrule[0.1pt]
    {\footnotesize |$\mathcal{E}$| } & {1,000,209} & {1,027,370} & {1,463,556} & {1,561,406} & {2,984,108} &{10,000,014} \\
   \midrule[0.1pt]
 Density  & {0.04190}   & {0.00084}   & {0.00267}   & {0.00130}  &{0.00062}  &{0.02210}  \\
\bottomrule
\end{tabular}}
\end{table}

\textbf{Baselines.}
\label{sec:baseline}
We include the following representative hashing-based models for (1) general object retrieval (LSH~\cite{lsh}), (2) image search (HashNet~\cite{hashnet}), and (3) Top-N candidate generation for recommendation (Hash\_Gumbel~\cite{gumbel1,gumbel2}, CIGAR~\cite{kang2019candidate} and HashGNN~\cite{hashgnn}).
We also include several state-of-the-art full-precision\footnote{They are denoted by FT32 as we implement them with float32 in the experiments.} recommender models, i.e., NeurCF~\cite{neurcf}, NGCF~\cite{ngcf}, DGCF~\cite{dgcf}, LightGCN~\cite{lightgcn}, for the long-list ranking quality comparison.
Model introductions are referred in Appendix D.
Early hashing methods, e.g., SH~\cite{weiss2008spectral}, RMMH~\cite{joly2011random}, LCH~\cite{zhang2010laplacian}, are excluded mainly because the above competing models~\cite{hashnet,kang2019candidate} have already validated the performance superiority over them.

%% file: exp_topn.tex
\begin{table*}[t]
\setlength{\abovecaptionskip}{0.2cm}
\setlength{\belowcaptionskip}{0.2cm}
\centering
\small
  \caption{Results of Recall@20 and NDCG@20 in Top-1000 retrieval: (1) ``R'' and ``N'' denote the Recall and NGCG; (2) the bold indicate \model~and the underline represents the best-performing models; (3) Mark \textbf{$^*$} denotes scenarios where Wilcoxon signed-rank tests indicate statistically significant improvements over the second-best models over 95\% confidence level.}
  \label{tab:topn}
  \setlength{\tabcolsep}{1mm}{
  \begin{tabular}{c|c c| c c| c c| c c| c c|c c} 
    \toprule
    Dataset & \multicolumn{2}{c|}{MovieLens (\%)} & \multicolumn{2}{c|}{Gowalla (\%)} & \multicolumn{2}{c|}{Pinterest (\%)} & \multicolumn{2}{c|}{Yelp2018 (\%)}  & \multicolumn{2}{c|}{AMZ-Book (\%)} & \multicolumn{2}{c}{Dianping (\%)} \\
    Metric & R@20$_{1000}$ & N@20$_{1000}$  & R@20$_{1000}$ & N@20$_{1000}$ & R@20$_{1000}$ & N@20$_{1000}$ & R@20$_{1000}$ & N@20$_{1000}$ & R@20$_{1000}$ & N@20$_{1000}$ & R@20$_{1000}$ & N@20$_{1000}$  \\ 
    \midrule[0.01pt]
    LSH                 & {11.38} & {25.87}  & {8.14} & {12.23}  & {7.88} & {6.71}  & {2.91} & {4.35}  & {2.41} & {2.34}  & {5.85} & {5.84}  \\
    HashNet             & {15.43} & {32.23}  & {11.38} & {13.74}  & {10.27} & {7.33}  & {3.37} & {4.41}  & {2.86} & {2.71}  & {6.24} & {5.59}  \\
    CIGAR               & {14.84} & {31.73}  & {11.57} & {14.21}  & {10.34} & {8.53}  & {3.65} & {4.57}  & {3.05} & {3.03}  & {6.91} & {6.03}  \\
    Hash\_Gumbel            & {16.62}  & {32.48} & {12.26}  & {14.68} & {10.53}  & {8.74} & {3.85}  & {5.12} & {2.69}  & {3.24} & {8.29}  & {6.43} \\
    HashGNN$_{\rm h}$   & {14.21} & {31.83}  & {11.63} & {14.21}  & {10.15} & {8.67}  & {3.77} & {5.04}  & {3.09} & {3.15}  & {8.34} & {6.68}  \\
    HashGNN$_{\rm s}$   & \underline{19.87} & \underline{33.21}  & \underline{13.45} & \underline{14.87}  & \underline{12.38} & \underline{9.11}  & \underline{4.86} & \underline{5.34}  & \underline{3.34} & \underline{3.45}  & \underline{9.57} & \underline{7.13}  \\

    \midrule[0.01pt]
    \textbf{\model} &\textbf{22.86$^*$} &\textbf{36.26$^*$} &\textbf{16.73$^*$} &\textbf{16.48$^*$} &\textbf{12.78$^*$} &\textbf{9.42$^*$} &\textbf{5.51$^*$} &\textbf{5.84$^*$} &\textbf{3.48$^*$} &\textbf{3.92$^*$} &\textbf{10.66$^*$} &\textbf{7.63$^*$}  \\
    \% Gain       &{ 15.05\%} &{ 9.18\%}    &{ 24.39\%} &{ 10.83\%}  &{ 3.23\%} &{ 3.40\%}  &{ 13.37\%} &{ 9.36\%}  &{ 4.19\%} &{ 13.62\%}  &{ 11.39\%} &{ 7.01\%} \\

   \bottomrule
  \end{tabular}}
\end{table*}

\begin{figure*}[t]
\begin{minipage}{1\textwidth}
\vspace{-0.05in}
\includegraphics[width=7.1in]{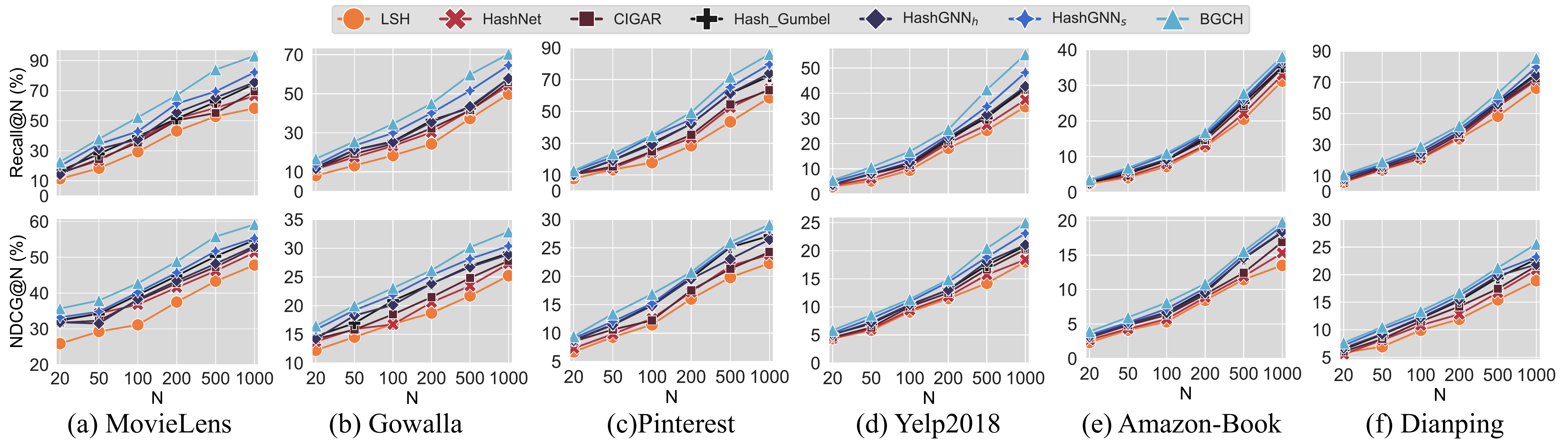}
\end{minipage} 
\vspace{-0.1in}
\caption{Top-N retrieval quality with N in \{20, 50, 100, 200, 500, 1000\} (best view in color).}
\label{fig:topn}
\end{figure*}

\subsection{Top-N Hamming Space Query (RQ1)}
\label{sec:exp_topn}

To evaluate \textbf{fine-to-coarse} Top-N ranking capability, we set N=1000.
We first report the results of Recall@20$_{1000}$ and NDCG@20$_{1000}$\footnote{We then use simple notation Recall@20, NDCG@20 if there is no ambiguity caused.} in Top-1000 search in Table~\ref{tab:topn} and then plot the holistic Recall and NDCG metric curves of \{20, 50, 100, 200, 500, 1000\} of Top-1000 in Figure~\ref{fig:topn}.
We set convolution iteration number as 2 and embedding dimension as 256 for \model~and baselines for fair comparison.
\begin{itemize}[leftmargin=*]

\item \textbf{The results demonstrate the superiority of \model~model over prior hashing-based models.}
(1) As shown in Table~\ref{tab:topn}, the state-of-the-art model, i.e., HashGNN, works better than traditional hashing-based baselines, e.g., LSH, HashNet, CIGAR. 
This indicates that, compared to graph-based models, a direct adaptation of conventional (i.e., non-graph-based) hashing methods may be hard to achieve comparable performance, mainly because of the effectiveness of \textit{graph convolutional} architecture in capturing latent information within the bipartite graph topology for hash encoding preparation.
(2) Owing to our proposed model components, e.g., \textit{adaptive graph convolutional hashing}, \model~consistently outperforms HashGNN over all datasets, by 3.23\%$\sim$24.39\%, and 3.40\%$\sim$ 13.62\% \textit{w.r.t.} Recall@20 and NDCG@20, respectively.
(3) Furthermore, we conduct the Wilcoxon signed-rank tests at \model.
The results verify that all \model~improvements over the second-best model are statistically significant over 95\% confidence level.
(4) To explain these, our proposed topology-aware graph convolutional hashing 
approach effectively enriches the graph node embeddings.
Our proposed feature dispersion further alleviates the feature erosion issue caused by numerical binarization. 
Last but not least, our proposed Fourier serialized gradient estimation is also vital to provide accurate gradients for model optimization.
We conduct the ablation study later in~\cref{sec:ablation}.

\item \textbf{By varying N from 20 to 1000, \model~consistently shows competitive performance compared to the baselines.}
While Recall@N indicates the fraction of relevant objects in Top-N retrieval, NDCG@N measures the ranking capability for relative orders.
As shown in Figure~\ref{fig:topn}:
(1) Compared to the approximated version of HashGNN, i.e., HashGNN$_{s}$, \model~generally obtains stable and significant improvements of both Recall and NDCG metrics over all six benchmarks with N from 20 to 1000.
(2) Apart from the higher retrieval quality, another advantage of \model~over HashGNN$_{s}$ is that it still supports bitwise operations, i.e., hamming distance matching, for inference acceleration. 
This is because, to improve the prediction accuracy, HashGNN$_{s}$ adopts a Bernoulli random variable to provide the probability of replacing the certain digits in the hash codes with the original continuous values, which thus disables the bitwise computation.
As we present in~\cref{sec:resource}, \model~achieves over 8$\times$ inference acceleration over HashGNN$_{s}$, which is particularly promising for query-based online matching and retrieval applications.
\end{itemize}

%% file: exp_full.tex
\subsection{\textbf{Comparing to FT32-based Models (RQ2)}}
\label{sec:exp_full}
\begin{table}[t]
\centering
\scriptsize
\caption{NDCG@1000 results of Float32-based models.}
\vspace{-0.15in}
\label{tab:full}
\setlength{\tabcolsep}{2mm}{
\begin{tabular}{c |c | c | c | c | c| c}
\toprule
    ~            & Movie & Gowalla & Pinterest & Yelp2018 & AMZ-Book & Dianping \\
\midrule
  NeurCF        & {58.76}  & {32.07} & {28.79} & {24.69} & {19.83} & {25.54} 	   \\
  NGCF         & {60.28}  & {32.13} & {29.78} & {25.23} & {20.37} & {25.76} 	   \\
  DGCF         & {62.41}  & {34.97} & \underline{31.47} & {26.28} & {21.74} & {26.87} 	   \\
  LightGCN     & \underline{62.88}  & \underline{35.26} & {31.32} & \underline{26.55} & \underline{21.92} & \underline{27.28} 	   \\
\midrule[0.1pt]
  \model            & \textbf{59.16}  & \textbf{32.87}   & \textbf{29.09}   & \textbf{25.01}   & \textbf{19.79}   & \textbf{25.57}   \\
  \% capacity	&{94.08\%} &{93.22\%} &{92.44\%} &{94.20\%} &{90.28\%} &{93.73\%} \\
\bottomrule
\end{tabular}}
\end{table}

In this section, we also compare \model~with several full-precision (FT32-based) models to evaluate the long-list search quality. 
As we can observe from Table~\ref{tab:full}, we have the following analyses.
(1) We notice that our model \model~generally performs competitively with early full-precision models, e.g., NeurCF and NGCF, over all datasets.
As for the state-of-the-art model LightGCN, our model can generally achieve over 90\% of the Top-1000 ranking capability.
(2) The performance of \model~demonstrates its effectiveness in guaranteeing the long-list Top-N retrieval quality.
This is useful for some industrial applications, e.g., recommender systems, which usually consist of two major stages: \textit{candidate generation} and \textit{re-ranking}.
Thus, obviously, the good quality of candidate generation directly reduces the complexity of next-stage re-ranking, as the search space is substantially pruned.
(3) Considering the \textit{efficiency in Hamming space retrieval} and the \textit{reduced space cost} of those learned hash codes, we believe that \model~can provide the optional alternative to these full-precision models, especially in scenarios with limited computation resources.

%% file: exp_ablation.tex
 \begin{table*}[t]
\centering
\footnotesize
\caption{Ablation study.}
\vspace{-0.15in}
\label{tab:ablation}
\setlength{\tabcolsep}{0.7mm}{
\begin{tabular}{c |c c|c c|c c|c c|c c|c c}
\toprule
 \multirow{2}*{Variant} & \multicolumn{2}{c|}{MovieLens} & \multicolumn{2}{c|}{Gowalla} & \multicolumn{2}{c|}{Pinterest} & \multicolumn{2}{c|}{Yelp2018}  &\multicolumn{2}{c|}{AMZ-Book} &\multicolumn{2}{c}{Dianping} \\
               ~  & R@20 & N@20 & R@20 & N@20 & R@20 & N@20 & R@20 & N@20 & R@20 & N@20 & R@20 & N@20\\
\midrule
\midrule
 \textsl{w/o FD}   &{22.82} {{\drop \notsotiny (-0.17\%)}} & {35.87} {{\drop \notsotiny (-1.08\%)}}  & {15.92} {{\drop \notsotiny (-4.84\%)}}  & {15.79} {{\drop \notsotiny (-4.19\%)}}  & {12.25} {{\drop \notsotiny (-4.15\%)}}  & {9.07} {{\drop \notsotiny (-3.72\%)}}  & {5.16} {{\drop \notsotiny (-6.35\%)}} & {5.49} {{\drop \notsotiny (-2.49\%)}}   &  {3.26} {{\drop \notsotiny (-6.32\%)}} & {3.57} {{\drop \notsotiny (-8.93\%)}}   &  {10.46} {{\drop \notsotiny (-1.88\%)}} &  {7.50} {{\drop \notsotiny (-1.70\%)}}    \\

 \textsl{w/o AH-TA}    &{19.54}{{\drop \notsotiny (-14.52\%)}} &{29.17}{{\drop \notsotiny (-19.55\%)}} &{13.49}{{\drop \notsotiny (-19.37\%)}} &{12.38}{{\drop \notsotiny (-24.88\%)}} &{12.24} {{\drop \notsotiny (-4.23\%)}} &{8.86} {{\drop \notsotiny (-5.94\%)}} &{4.77}{{\drop \notsotiny (-13.43\%)}} &{5.18}{{\drop \notsotiny (-11.30\%)}} &{2.49}{{\drop \notsotiny (-28.45\%)}} &{2.86}{{\drop \notsotiny (-27.04\%)}} &{\ \ 9.83} {{\drop \notsotiny (-7.79\%)}} &{6.87} {{\drop \notsotiny (-9.96\%)}}  \\

\textsl{w/o AH-RF}   &{16.73}{{\drop \notsotiny (-26.82\%)}} &{26.97}{{\drop \notsotiny (-25.62\%)}} &{11.24}{{\drop \notsotiny (-32.82\%)}} &{11.29}{{\drop \notsotiny (-31.49\%)}} &{10.18}{{\drop \notsotiny (-20.34\%)}} &{7.33}{{\drop \notsotiny (-22.19\%)}} &{3.76}{{\drop \notsotiny (-31.76\%)}} &{4.30}{{\drop \notsotiny (-26.37\%)}} &{3.27} {{\drop \notsotiny (-6.03\%)}} &{3.64} {{\drop \notsotiny (-7.14\%)}} &{\, \ 8.33}{{\drop \notsotiny (-21.86\%)}} &{6.93} {{\drop \notsotiny (-9.17\%)}}  \\
\midrule[0.1pt]

 \textsl{w/in LF}  &{21.06} {{\drop \notsotiny (-7.87\%)}} &{34.59} {{\drop \notsotiny (-4.61\%)}} &{15.48} {{\drop \notsotiny (-7.47\%)}} &{15.38} {{\drop \notsotiny (-6.67\%)}} &{11.94} {{\drop \notsotiny (-6.57\%)}} &{8.89} {{\drop \notsotiny (-5.63\%)}} &{4.86}{{\drop \notsotiny (-11.80\%)}} &{5.17}{{\drop \notsotiny (-11.47\%)}} &{3.14} {{\drop \notsotiny (-9.77\%)}} &{3.62} {{\drop \notsotiny (-7.65\%)}} &{\, \ 9.40}{{\drop \notsotiny (-11.82\%)}} &{7.27} {{\drop \notsotiny (-4.72\%)}}  \\

 \textsl{w/o $\mathcal{L}_{bpr}$} &{21.42} {{\drop \notsotiny (-6.30\%)}} &{34.83} {{\drop \notsotiny (-3.94\%)}} &{15.87} {{\drop \notsotiny (-5.14\%)}} &{15.66} {{\drop \notsotiny (-4.98\%)}} &{12.33} {{\drop \notsotiny (-3.52\%)}} &{9.17} {{\drop \notsotiny (-2.65\%)}} &{5.31} {{\drop \notsotiny (-3.63\%)}} &{5.61} {{\drop \notsotiny (-3.94\%)}} &{3.35} {{\drop \notsotiny (-3.74\%)}} &{3.77} {{\drop \notsotiny (-3.83\%)}} &{10.21} {{\drop \notsotiny (-4.22\%)}} &{7.38} {{\drop \notsotiny (-3.28\%)}}  \\

 \textsl{w/o $\mathcal{L}_{rec}$}  &{17.01}{{\drop \notsotiny (-25.59\%)}} &{27.16}{{\drop \notsotiny (-25.10\%)}} &{12.27}{{\drop \notsotiny (-26.66\%)}} &{12.63}{{\drop \notsotiny (-23.36\%)}} &{10.81}{{\drop \notsotiny (-15.41\%)}} &{7.86}{{\drop \notsotiny (-16.56\%)}} &{3.93}{{\drop \notsotiny (-28.68\%)}} &{4.37}{{\drop \notsotiny (-25.17\%)}} &{3.19} {{\drop \notsotiny (-8.33\%)}} &{3.73} {{\drop \notsotiny (-4.85\%)}} &{\, \ 8.82}{{\drop \notsotiny (-17.26\%)}} &{7.26} {{\drop \notsotiny (-4.85\%)}}  \\

\midrule[0.1pt]
  \textbf{\model}   &\textbf{22.86}& \textbf{36.26}  &\textbf{16.73}& \textbf{16.48} &\textbf{12.78}& \textbf{9.42} &\textbf{5.51} & \textbf{5.84} &\textbf{3.48} & \textbf{3.92} &\textbf{10.66}& \textbf{7.63}    \\ 

\bottomrule
\end{tabular}}
\end{table*}

\subsection{\textbf{Ablation Study (RQ3)}}
\label{sec:ablation}
We evaluate the necessity of model components with Top-20 search metrics and report the results in Table~\ref{tab:ablation}.

{\textbf{Effect of Feature Dispersion.}}
We first analyze the effect of our proposed feature dispersion approach for hedging the feature erosion in hash encoding.
We introduce the model variant, denoted by \textsl{w/o FD}, to directly disable it by setting $\eta$ as 0.
As shown in Table~\ref{tab:ablation}, the performance gap between \textsl{w/o FD} and \model~ well demonstrates the effectiveness of dispersing the latent features before embedding binarization for hashing over these six datasets. 
Moreover, let the density summarized in Table~\ref{tab:datasets} be computed by $\frac{|\mathcal{V}_1|\times |\mathcal{V}_2|}{|\mathcal{E}|}$. 
In sparse datasets, i.e., Gowalla (0.00084), Pinterest (0.00267), Yelp2018 (0.00130), and AMZ-Book (0.00062), the performance decay between \model~and \textsl{w/o FD} is much larger than on the other two datasets, i.e., MovieLens (0.04190) and Dianping (0.02210).
This is because sparse datasets are more sensitive to hashing as they may not have insufficient training edges to abridge the gap against their unhashed version.
Another promising approach to tackle data sparsity issue is \textit{data augmentation}~\cite{zhang2022costa} and we leave it for future work.

{\textbf{Effect of Adaptive Graph Convolutional Hashing.}}
Then we study this model component by setting two variants, where: (1) \textsl{w/o AH-TA} only disables the \textit{topology-awareness of hashing} and sets it as the final encoder after all graph convolutions (just like the conventional manner~\cite{hashgnn,hashnet}); (2) \textsl{w/o AH-RF} removes the \textit{rescaling factors}.
From Table~\ref{tab:ablation} results, we have the following observations:
\begin{enumerate}[leftmargin=*]
\item 
The variant \textsl{w/o AH-TA} consistently underperforms \model.
This demonstrates that simply using the rear output embeddings from the GCN framework may not sufficiently model the unique latent node features for hashing, especially for the rich structural information within different graph depths.
While in \model, by capturing the intermediate information for representation enrichment, the topology-aware hashing can effectively alleviate the limited expressivity of discrete hash codes.

\item Apart from the topology-aware hashing, another key point for contributing to the performance improvement is the \textit{rescaling factor} that we introduced in Equation~(\ref{eq:rescale}).
After removing it from \model, variant \textsl{w/o AH-RF} presents huge performance decay.
Although these factors are directly calculated and may not be theoretically optimal, they reflect the numerical uniqueness of embeddings for later hash encoding, which substantially improves \model's prediction capability. 
We study the \textit{determinacy} design of factor computation in the following section.
\end{enumerate}

{\textbf{Design of Learnable Rescaling.}}
We include another variant namely \textsl{w/in LF} to indicate the model version using \textit{learnable rescaling factors}. 
As shown in Table~\ref{tab:ablation}, the design of learnable rescaling factors in \textsl{w/in LF} does not achieve good performance as expected. 
One explanation is that, our proposed model currently does not post a strong mathematical constraint to the learnable factors ($\alpha_x$), e.g., $\alpha_x^{(l)} = \argmin(\widetilde{\emb{V}}^{(l)}_x$, $\alpha_x^{(l)} \emb{{Q}}^{(l)}_x)$, mainly because of its additional training complexity; and purely relying on the stochastic optimization, e.g., stochastic gradient descent (SGD), may hardly reach the optimum.
Considering the additional search space introduced from this regularization design, we argue that our deterministic rescaling method is simple yet effective in practice.

{\textbf{Effect of Multi-loss in Optimization.}}
Lastly, to study the effect of BPR loss $\mathcal{L}_{bpr}$ and graph reconstruction loss $\mathcal{L}_{rec}$, we set two variants, termed by \textsl{w/o $\mathcal{L}_{bpr}$} and \textsl{w/o $\mathcal{L}_{rec}$}, to optimize \model~separately.
As shown in Table~\ref{tab:ablation}, with all other model components, partially using each one of $\mathcal{L}_{bpr}$ and $\mathcal{L}_{rec}$ can not achieve the expected performance.
This confirms the effectiveness of our proposed multi-loss design:
while $\mathcal{L}_{bpr}$ learns to assign higher prediction values to observed edges, i.e., $\emb{Y}_{x,y}=1$, than the unobserved node pair counterparts, 
$\mathcal{L}_{rec}$ transfers the graph reconstruction problem to a classification task by using the original embeddings in training.
By collectively optimizing these two loss functions, our model \model~can learn precise intermediate embeddings from $\mathcal{L}_{rec}$, and generate targeted hash codes with high-quality relative order information regularized by $\mathcal{L}_{bpr}$ accordingly.

%% file: exp_resource.tex
\subsection{\textbf{Resource Consumption Analysis (RQ4)}}
\label{sec:resource}
Due to the various value ranges over all six datasets, we compactly report the value ratios of \model~over the state-of-the-art hashing-based model HashGNN$_s$ in Figure~\ref{fig:tradeoff}.

\begin{figure}[h]
\begin{minipage}{0.5\textwidth}
\includegraphics[width=3.4in]{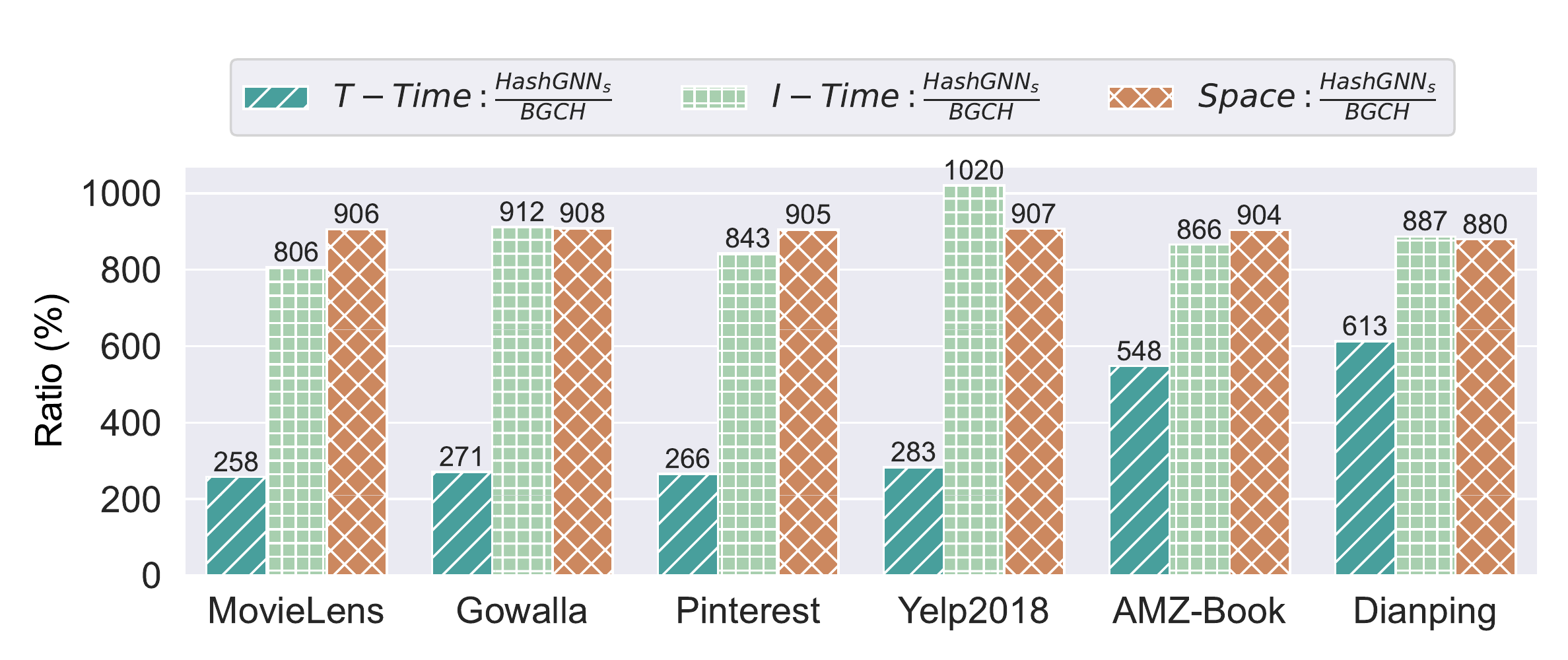}
\end{minipage} 
\vspace{-0.15in}
\caption{Resource consumption ratios.}
\label{fig:tradeoff}
\end{figure}

{\textbf{Model Training Time Cost.}}
As indicated by the metric ``\textit{T-Time}'' in Figure~\ref{fig:tradeoff}, we notice that training HashGNN$_s$ is more time-consuming than our proposed model.
The main reason is that HashGNN adopts the early GCN framework~\cite{graphsage} as the model backbone, while our model follows the latest framework~\cite{lightgcn} to remove operations, e.g., self-connection, feature transformation, and nonlinear activation.
In addition, on the two largest datasets AMZ-Book and Dianping, the training cost ratio further increases to around 5$\sim$6 times.
This is because we have to decrease the batch size of HashGNN$_s$ for tractable training process.

{\textbf{Online Inference Time Cost.}}
We randomly generate 1,000 queries and evaluate the computation time cost.
To present a fair comparison, we disable all parallel arithmetic techniques (e.g., MKL, BLAS) by using the open-source toolkit\footnote{\url{https://www.lfd.uci.edu/~gohlke/pythonlibs/}}.
Indicated by ``\textit{I-Time}'' in Figure~\ref{fig:tradeoff}, our model with Hamming distance matching generally achieves over 8$\times$ computation acceleration over HashGNN$_s$ on all datasets.
This is because, as we have explained in~\cref{sec:exp_topn}, HashGNN$_s$ randomly replaces the hash codes with their original continuous embeddings for relaxation and adopts floating-point arithmetics to pursue performance improvement while sacrificing the computation acceleration from the bitwise operations.

{\textbf{Hash Codes Memory Footprint.}}
Binarized embeddings can largely reduce memory space consumption.
Compared to the state-of-the-art hashing-based model HashGNN$_s$, our \model~further achieves about 9$\times$ of memory space reduction for the hash codes.
As we have just explained, since HashGNN$_s$ interprets hash codes with random real-value digits, it thus requires additional cost to distinguish binary digits from full-precision ones. 
On the contrary, \model~ separates the storage of binarized encoding parts and corresponding rescaling factors, thus providing the advantage for space overhead optimization.

%% file: exp_fs.tex
\subsection{\textbf{Study of Fourier Gradient Estimation (RQ5)}}
\label{sec:fs_exp}
We take our largest dataset Dianping for illustration and the analysis can be generally popularized to the other datasets. 

{\textbf{Effect of Decomposition Term $n$.}}
We vary the decomposition term $n$ from 1 to 16.
As shown in Figure~\ref{fig:fs_n}, we have two observations:
(1) Different decomposition terms will surely affect the final retrieval quality, as theoretically, the larger $n$ increases, the more accurate gradients can be estimated.
However, in practice, too large values of $n$ may introduce the overfitting risk, which implies that keeping a moderate $n$, e.g., $n$=4 in Figure~\ref{fig:fs_n}(a), can already maximize the model performance.
(2) By varying $n$ from 1 to 16, the training time per iteration of \model~ slowly increases. 
This generally coincides with our complexity analysis in Appendix C, in which the majority of training cost lies in our feature dispersion and graph convolutional hashing, as $O(\frac{2cs(K+L)|\mathcal{E}|^2}{B}) \gg O(snd|\mathcal{E}|)$.

{\textbf{Comparison with Other Gradient Estimators.}}
We include several recent gradient estimators, i.e., \textit{Tanh-like}~\cite{qin2020forward,gong2019differentiable}, \textit{SignSwish}~\cite{darabi2018bnn}, \textit{Sigmoid}~\cite{sigmoid}, and \textit{projected-based estimator}~\cite{RBCN} (denoted as PBE).
(1) The results summarized in Table~\ref{tab:estimator} well demonstrate the superiority of our proposed Fourier Series decomposition to $\sign(\cdot)$ function in gradient estimation.
As we have briefly explained, most existing estimators employ the \textit{visually similar} function approximation to $\sign(\cdot)$; compared to STE, they generally provide better gradient estimation.
(2) However, for those bipartite graphs with heavy sparsity, e.g., Gowalla (0.00084) and AMZ-Book (0.00062), graph-based models may hardly collect enough structural information for effective hash codes training.
Based on the limited training samples, these \textit{theoretically irrelevant} estimators may not effectively rectify the optimization deviation, and thus present a recognizable performance gap against our proposed Fourier serialized estimator.

\begin{figure}[t]
\begin{minipage}{0.5\textwidth}
\includegraphics[width=3.3in]{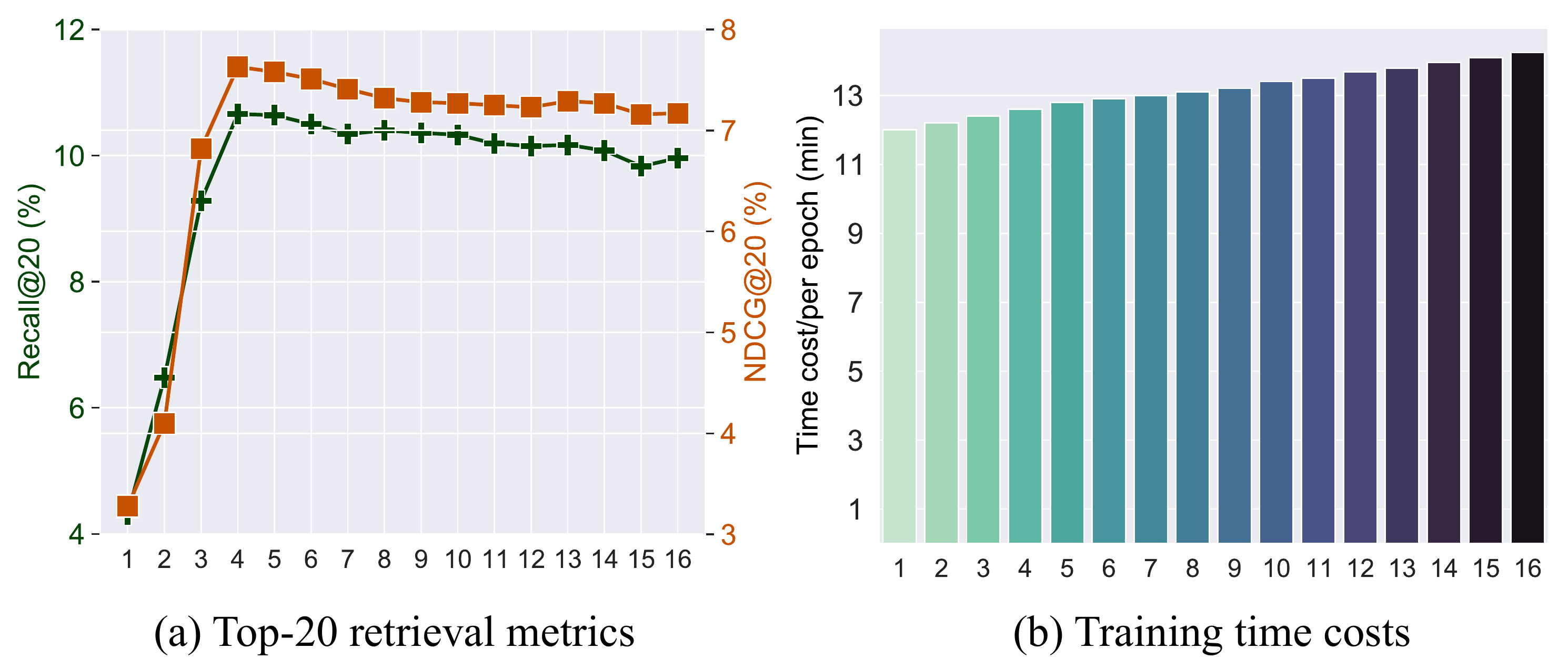}
\end{minipage} 
\vspace{-0.1in}
\caption{Fourier Series decomposition term $n$ in \model.}
\label{fig:fs_n}
\end{figure}

\begin{table}[t]
\centering
\scriptsize
\caption{Gradient estimator comparison on Recall@20.}
\vspace{-0.1in}
\label{tab:estimator}
\setlength{\tabcolsep}{0.7mm}{
\begin{tabular}{c |c | c | c | c | c| c}
\toprule
    ~            & Movie & Gowalla & Pinterest & Yelp2018 & AMZ-Book & Dianping \\
\midrule
  STE         & {20.93 ({\drop \tiny -8.44\%})}   & {14.85({\drop \tiny -11.24\%})}   & {12.35 ({\drop \tiny -3.36\%})}   & {5.24 ({\drop \tiny -4.90\%})}   & {3.12({\drop \tiny -10.34\%})}   & {10.34 ({\drop \tiny -3.00\%})}  \\
  Tanh         & {21.75 ({\drop \tiny -4.86\%})}   & {15.06 ({\drop \tiny -9.98\%})}   & {12.36 ({\drop \tiny -3.29\%})} & {5.43 ({\drop \tiny -1.45\%})}   & {3.21 ({\drop \tiny -7.76\%})}     & {10.41 ({\drop \tiny -2.34\%})}  \\
  SignSwish         & {22.13 ({\drop \tiny -3.19\%})}   & {15.62 ({\drop \tiny -6.63\%})}   & {12.44 ({\drop \tiny -2.66\%})}   & {5.50 ({\drop \tiny -0.18\%})}   & {3.34 ({\drop \tiny -4.02\%})}   & {10.43 ({\drop \tiny -2.16\%})}  \\
  Sigmoid         & {22.08 ({\drop \tiny -3.41\%})}   & {15.21 ({\drop \tiny -9.09\%})}   & {12.52 ({\drop \tiny -2.03\%})}   & {5.53 ({\drop \tiny +0.03\%})}   & {3.18 ({\drop \tiny -8.62\%})}   & {10.38 ({\drop \tiny -2.63\%})}  \\
  PBE         & {21.68 ({\drop \tiny -5.16\%})}   & {15.05({\drop \tiny -10.04\%})}   & {12.32 ({\drop \tiny -3.60\%})}   & {5.35 ({\drop \tiny -2.90\%})}   & {3.13({\drop \tiny -10.06\%})}   & {10.47 ({\drop \tiny -1.78\%})}  \\
\midrule[0.1pt]
  \model            & \textbf{22.86}  & \textbf{16.73}   & \textbf{12.78}   & \textbf{5.51}   & \textbf{3.48}   & \textbf{10.66}   \\
\bottomrule
\end{tabular}}
\end{table}

%% file: con.tex
\section{\textbf{Conclusion}}
\label{sec:con}
We study the graph convolutional hashing over bipartite graphs for efficient Hamming space search, by proposing \model~with three effectual modules.
Extensive experiments demonstrate the model superiority over conventional counterparts and validate the effectiveness of all proposed modules. 
As for future work, we plan to investigate modeling with the \textit{semi-supervised} graph setting~\cite{zixing1,zixing2} mainly for its commonality in practice. 
Moreover, another promising direction is to upgrade \model~ for \textit{inductive learning}~\cite{graphsage}, such that it can make adaptive matching and prediction for evolving graphs with structural updates.


%% file: app.tex
\appendix

\section{Loss Landscape Visualization}
\label{sec:visualization}
We simulate the optimization trajectories of learnable embeddings and visually compare the loss landscapes of non-hashing and hashing versions in Figure~\ref{fig:model}(a).
Specifically, we manually assign perturbations~\cite{nahshan2021loss, bai2020binarybert} to the embeddings on MovieLens dataset as: {\footnotesize$\emb{V}_x^{(l)} = \emb{V}_x^{(l)} \pm p \cdot$ $\overline{|{\emb{V}_x^{(l)}|}}$ $\cdot \emb{1}^{(l)}$}.
where {\footnotesize$\overline{|{\emb{V}_x^{(l)}|}}$} represents the absolute mean of entries in {\footnotesize$\emb{V}_x^{(l)}$} and perturbation magnitudes $p$ are from $\{0.01, \cdots, 0.50\}$. $\emb{1}$ is an all-one vector. 
For pairs of perturbed node embeddings, we plot their loss distribution accordingly.
As we can observe, the non-hashing version produces a flat loss surface, showing the local convexity.
On the contrary, the hashing counterpart has a bumping and complex loss landscape.

\section{Notation Table and \model~Pseudo-codes}
\label{app:notation_and_code}
We use bold uppercase and calligraphy characters for matrices and sets. The non-bolded denote graph nodes or scalars. 
Key notations and Pseudocodes are explained in Table~\ref{tab:notation} and Algorithm~\ref{alg:model}.

\begin{table}[t]
\caption {Notations and meanings.}
\vspace{-0.15in}
\label{tab:notation}
  \footnotesize
  \begin{tabular}{c|l} 
     \hline
          {\bf Notation} & {\bf Meaning}\\
     \hline\hline
          {\notsotiny $\mathcal{G},\mathcal{V}_1$, $\mathcal{V}_2$, $\mathcal{E}$} & Bipartite graph with sets of nodes and edges.\\
    \hline
        {$c$, $d$}  & Convolution dimension and hash code dimension.\\
    \hline
        {$\emb{V}_x^{(l)}$}  & {Node $x$'s embedding at iteration $l$.} \\
    \hline
         \tabincell{l}{$\emb{Y}$}   & \tabincell{l}{{Edge transactions where $\emb{Y}_{x,y}=1$ indicates the interaction} \\ existence between nodes $x$ and $y$, and otherwise $\emb{Y}_{x,y}=0$.} \\
    \hline
        {$\emb{A}$}, $\emb{D}$    & {Adjacency matrix and associated diagonal degree matrix.}  \\
    \hline
        $\emb{p}^{(k)}$, $\emb{P}$ & Dispersing vector at iteration $k$ and the projection matrix.\\
    \hline
        $\widetilde{V}$ & Feature-dispersed embedding matrix. \\
    \hline
        {$\widehat{\emb{Y}}$} & {Estimated matching scores.} \\
    \hline
        {$\emb{Q}_x^{(l)}$}   &  {Hash code segment of node $x$ at iteration $l$.} \\
    \hline
        $\alpha^{(l)}$  & $x$'s rescaling factor computed at the $l$-th convolution.\\
    \hline
        {$\emb{Q}_x$}  &   {Target hash codes of node $x$.} \\
    \hline
        {$L$}, {$K$}  &  {Numbers of convolutional hashing and dispersion generation.}\\
    \hline
        $\mathcal{L}_{rec}$, $\mathcal{L}_{bpr}$, $\mathcal{L}$ & {Two loss terms of final objective function $\mathcal{L}$.} \\
    \hline
      {$\eta$, $H$, $n$, $\lambda_1$, $\lambda_2$}  & hyper-parameters.\\
    \hline
  \end{tabular}
\end{table}

\begin{algorithm}[t]
\small
\caption{\model~algorithm.}
\label{alg:model}
\LinesNumbered  
\While{\rm{model not converge}}{
	\For{$k = 0, \cdots, K-1$}{
		$\emb{p}^{(k+1)}$ $\gets$ $(\emb{V}^{(0)})^\mathsf{T}\emb{V}^{(0)}\emb{p}^{(k)}$; \\
	}
	$\emb{P} \gets$ obtain the projection matrix \Comment*[r]{Eq.(\ref{eq:projection})} 
	$\widetilde{\emb{V}}^{(0)} \gets$ obtain the feature-dispersed embeddings \Comment*[r]{Eq.(\ref{eq:disperse})}
    \For{$l = 0, \cdots, L-1$}{
          $\widetilde{\emb{V}}^{(l+1)} \gets(\emb{D}^{-\frac{1}{2}} \emb{A} \emb{D}^{-\frac{1}{2}} )\widetilde{\emb{V}}^{(l)}$ \Comment*[r]{Eq.(\ref{eq:fdconv})}
          $\emb{{Q}}^{(l+1)} \gets \sign\big(\widetilde{\emb{V}}^{(l+1)})$ \Comment*[r]{Eq.(\ref{eq:hashing})}
          $\emb{\alpha} \gets$ calculate the rescaling factors \Comment*[r]{Eq.(\ref{eq:rescale})}
        }
      \For{$x \in \mathcal{V}_1, y \in \mathcal{N}(x)$}{
      $\widehat{\emb{Y}}_{x,y} \gets$ $\alpha_x\alpha_y$ $(d - 2D_{H}(\emb{Q}_x, \emb{Q}_y))$ \Comment*[r]{Eq.(\ref{eq:inner_score})\&Thm.\ref{tm:equal}}
     }
      $\mathcal{L} \gets$ compute loss and optimize the model \Comment*[r]{Eq's.(\ref{eq:rec}-\ref{eq:L})} 
      
}
\textbf{Function} \tt{Gradient\_estimator}($\mathcal{L}$): \\
$\frac{\partial \mathcal{L}}{\partial \emb{V}} \gets \frac{\partial \mathcal{L}}{\partial \emb{Q}} \cdot \frac{4}{H} \sum_{i=1,3,5,\cdots}^{n} \cos(\frac{\pi i \emb{V}}{H})$ \Comment*[r]{Eq.(\ref{eq:gradient})}
\end{algorithm}

\input{analysis}

\section{Experiment Setup Details}
\label{app:exp}
{\textbf{Datasets.}}
We evaluate our model on the following six six datasets:
\begin{enumerate}[leftmargin=*]
\item \textbf{MovieLens}\footnote{\url{https://grouplens.org/datasets/movielens/1m/}} is a widely adopted benchmark between \textit{users} and \textit{movies}. Similar to the setting in~\cite{hashgnn,he2016fast}, if the user $x$ has rated item $y$, we set the edge $\emb{Y}_{x,y} = 1$, otherwise $\emb{Y}_{x,y} = 0$. 
\item \textbf{Gowalla}\footnote{\url{https://github.com/gusye1234/LightGCN-PyTorch/tree/master/data/gowalla}}~\cite{ngcf,hashgnn,lightgcn,dgcf} is the dataset~\cite{liang2016modeling} between \textit{customers} and \textit{their check-in locations} collected from Gowalla. 
\item \textbf{Pinterest}\footnote{\url{https://sites.google.com/site/xueatalphabeta/dataset-1/pinterest_iccv}} is an open dataset for image recommendation between \textit{users} and \textit{images}.
Edges represent the pins over images initiated by users. 
\item \textbf{Yelp2018}\footnote{\url{https://github.com/gusye1234/LightGCN-PyTorch/tree/master/data/yelp2018}} is from Yelp Challenge 2018 Edition, bipartitely modeling between \textit{users} and \textit{local businesses}.
\item \textbf{AMZ-Book}\footnote{\url{https://github.com/gusye1234/LightGCN-PyTorch/tree/master/data/amazon-book}} is the bipartite graph between \textit{readers} and \textit{books}, organized from the book collection of Amazon-review~\cite{he2016ups}.  
\item \textbf{Dianping}\footnote{\url{https://www.dianping.com/}} is a commercial dataset between \textit{users} and \textit{local businesses} recording their diverse interactions, e.g., clicking, saving, and purchasing. 
\end{enumerate}

\textbf{Evaluation metrics.}
To evaluate the model performance of Hamming space retrieval over bipartite graphs, we directly deploy our model \model~ in the basic user-item recommendation scenarios.
Specifically, given a query node, we apply the hash codes to match Top-N answers for the query with the closest Hamming distances, and thus adopt two widely-used evaluation protocols Recall@N and NDCG@N to measure the ranking capability.

\vspace{0.05in}

{\textbf{Implementations.}}
We implement our models under Python 3.6 and PyTorch 1.14.0 on a Linux machine with 1 Nvidia GeForce RTX 3090 GPU, 4 Intel Core i7-8700 CPUs, 32 GB of RAM with 3.20GHz.
For all the baselines, we follow the official hyper-parameter settings.
We apply a grid search if lacking recommended model settings.
The dimension is searched in \{$32, 64, 128, 256, 512$\}. 
The learning rate $\eta$ is tuned within \{$10^{-3}, 5\times10^{-3}, 10^{-2}, 5\times10^{-2}$\} and the coefficient $\lambda$ is tuned among \{$10^{-5}, 10^{-4}, 10^{-3}$\}. 
We initialize and optimize all models with default normal initializer and Adam optimizer~\cite{adam}. 

\vspace{0.05in}

\textbf{Baselines.} All baselines studied in this paper are introduced as:
\label{app:baselines}
\begin{enumerate}[leftmargin=*]
\item \textbf{LSH}~\cite{lsh} is a classical hashing method. LSH is proposed to approximate the similarity search for massive high-dimensional data and we introduce it for Top-N object search by following the adaptation in~\cite{hashgnn}. 

\item \textbf{HashNet}~\cite{hashnet} is a representative deep hashing method that is originally proposed for multimedia retrieval tasks.
Similar to~\cite{hashgnn}, we adapt it for graph data by modifying it with the general graph convolutional network.

\item \textbf{CIGAR}~\cite{kang2019candidate} is a state-of-the-art neural-network-based framework for fast Top-N candidate generation in recommendation. 
CIGAR can be further followed by a full-precision re-ranking algorithm. And we only use its hashing part for fair comparison.

\item \textbf{Hash\_Gumbel} is a variance of \model~with Gumbel-softmax for hash encoding and gradient estimation~\cite{gumbel1,gumbel2}.
Specifically, we first expand each embedding bit to a size-two one-hot encoding. 
Then it utilizes the Gumbel-softmax trick to replace $\sign(\cdot)$ as relaxation for binary hash code generation. 

\item \textbf{HashGNN}~\cite{hashgnn} is the state-of-the-art learning to hash based method with GCN framework. 
We use HashGNN$_{h}$ to denote the vanilla version with \textit{hard encoding} proposed in~\cite{hashgnn}, where each element of user-item embeddings is strictly binarized. 
We use HashGNN$_{s}$ to denote its proposed approximated version.

\item \textbf{NeurCF}~\cite{neurcf} is one representative deep neural network model for collaborative filtering in recommendation. 

\item \textbf{NGCF}~\cite{ngcf} is one of the representative graph-based recommender models with collaborative filtering methodology. 

\item \textbf{DGCF}~\cite{dgcf} is a state-of-the-art graph-based model that learns disentangled user intents for better recommendation. 

\item \textbf{LightGCN}~\cite{lightgcn} is another latest state-of-the-art GCN-based recommender model that has been widely evaluated. 

\end{enumerate}

%% file: analysis.tex
\section{Theoretical Proofs and Analyses}
\label{sec:discuss}
\addtocounter{thm}{-2}

\input{fd_analysis}

\begin{thm}[\textbf{Hamming Distance Matching}]
Given two hash codes, we have $(\alpha_x\emb{Q}_u)^\mathsf{T} \cdot (\alpha_y\emb{Q}_y)$ $=$ $\alpha_x\alpha_y$ $(d - 2D_{H}(\emb{Q}_x, \emb{Q}_y))$.
\end{thm} 

\begin{proof}
\begin{equation}
\begin{aligned}
\setlength\abovedisplayskip{2pt}
\setlength\belowdisplayskip{2pt}
\emb{Q}_x^\mathsf{T} \cdot \emb{Q}_y &= \big|\{ i|(\emb{Q}_{x})_i = (\emb{Q}_{y})_i = 1\}\big| +  \big|\{ i|(\emb{Q}_{x})_i = (\emb{Q}_{y})_i = -1\}\big| \\ 
& -  \big|\{ i|(\emb{Q}_{x})_i \neq (\emb{Q}_{y})_i\}\big|\\
& = d - 2 \cdot \big|\{ i|(\emb{Q}_{x})_i \neq (\emb{Q}_{y})_i\}\big|  = \underline{d - 2D_{H}(\emb{Q}_x, \emb{Q}_y))},\\
\end{aligned}
\end{equation}%
which completes the proof.
\end{proof}

\textbf{Training time complexity.}
As shown in Table~\ref{tab:train_time}, $|\mathcal{E}|$, $B$, $s$, and $n$ are the edge number, batch size, numbers of train iterations and Fourier Series decomposition terms.
(1) The time complexity of the graph normalization, i.e., $\emb{D}^{-\frac{1}{2}} \emb{A} \emb{D}^{-\frac{1}{2}}$, is $O(2|\mathcal{E}|)$.
(2) Before the graph convolution, we first conduct the feature dispersion only for the initial node embeddings, e.g., $\emb{V}_x^{(0)}$, which takes $O(\frac{2csK|\mathcal{E}|^2}{B})$ complexity.
In our experiment, hyper-parameter $K \leq 3$.
(3) In graph convolution, the time complexity is $O(\frac{2csL|\mathcal{E}|^2}{B})$, where $L \leq 4$ is a common setting~\cite{lightgcn,ngcf,kipf2016semi,graphsage} to avoid \textit{over-smoothing}~\cite{li2019deepgcns}.
(4) As for the loss computation, \model~takes $O\big(2sc|\mathcal{E}|\big)$ to compute $\mathcal{L}_{rec}$ and $O\big(2sd|\mathcal{E}|\big)$ for $\mathcal{L}_{bpr}$.
(5) Lastly, \model~takes $O(snd|\mathcal{E}|)$ to estimate the gradients for the $d$-dimension hash codes.
Thus, thee total complexity in total is quadratic to the graph edge numbers, i.e., $|\mathcal{E}|$, which is common in GCN frameworks.

 \begin{table}[t]
\setlength{\abovecaptionskip}{0.2cm}
\setlength{\belowcaptionskip}{0.2cm}
\centering
\notsotiny
\caption{Traing time complexity.}
\vspace{-0.05in}
\label{tab:train_time}
\setlength{\tabcolsep}{2mm}{
  \begin{tabular}{c | c | c | c | c }
\toprule
    {Graph Norm.}  & {Feat. Disp.}  & {Conv. \& Hash.}   & { Loss Comp.}  &{Grad. Est.} \\
\midrule[0.1pt]
    { $O(2|\mathcal{E}|)$} & {$O(\frac{2csK|\mathcal{E}|^2}{B})$} & {$O(\frac{2csL|\mathcal{E}|^2}{B})$} & { $O\big(2s(c$+$d)|\mathcal{E}|\big)$} & { $O(snd|\mathcal{E}|)$} \\
\bottomrule
\end{tabular}}
\end{table}

\begin{table}[t]
\setlength{\abovecaptionskip}{0.2cm}
\setlength{\belowcaptionskip}{0.2cm}
\centering
\footnotesize
\caption{Complexity of space and computation.}
\vspace{-0.05in}
\label{tab:prediction}
\setlength{\tabcolsep}{2.2mm}{
\begin{tabular}{c | c | c c}
\toprule
  ~          & {\footnotesize Space cost}   &  {\footnotesize \#FLOP} & {\footnotesize \#BOP}       \\
\midrule
\midrule
 {\scriptsize float32-based}      & {\scriptsize $32|\mathcal{V}_1 \cup \mathcal{V}_2|d$}       &   {\scriptsize $O\big(|\mathcal{V}_1| \cdot |\mathcal{V}_2| d\big)$}      &   {-}         \\
\midrule[0.1pt]
{\scriptsize \model}       & {\scriptsize $|\mathcal{V}_1 \cup \mathcal{V}_2| (d+32(L+1))$}    & {\scriptsize $O\big(4|\mathcal{V}_1| \cdot |\mathcal{V}_2| \big)$}            & {\scriptsize $O\big(|\mathcal{V}_1| \cdot |\mathcal{V}_2| d\big)$}    \\
\bottomrule
\end{tabular}}
\end{table}


\textbf{Hash codes space cost.}
As shown in Table~\ref{tab:prediction}, the total space cost of hash codes is {\small $O(|\mathcal{V}_1 \cup \mathcal{V}_2| (d+32(L+1)))$} bits, supposing that we use float32 for those rescaling factors in $L+1$ iterations.
Compared to the continuous embedding size, i.e., $32|\mathcal{V}_1 \cup \mathcal{V}_2|d$ bits, the theoretical size reduction ratio thus is:
\begin{equation}
\setlength\abovedisplayskip{2pt}
\setlength\belowdisplayskip{2pt}
\label{eq:space}
ratio = \frac{32|\mathcal{V}_1 \cup \mathcal{V}_2|d}{|\mathcal{V}_1 \cup \mathcal{V}_2| (d+32(L+1))} = \frac{32d}{d+32(L+1)}.
\end{equation}
As we just explained, stacking too many iteration layers will incurring performance detriment~\cite{li2019deepgcns}. Hence, if $L\leq4$ and $d=1024$, \model~can achieve considerable size compression. 

\vspace{0.05in}

{\textbf{Online matching.}}
The original score formulation in Equation~(\ref{eq:inner_score}) contains $d$ floating-point operations (\#FLOPs).
As shown in Table~\ref{tab:prediction}, using Hamming distance matching can convert the most of floating-point arithmetics to binary operations (\#BOPs), with slightly more \#FLOPs for scalar computations, i.e., $4\ll d$.

%% file: fd_analysis.tex
\begin{thm}[\textbf{Feature Dispersion}]
Let ${\emb{V}}^{(l)} = \emb{U}_1\emb{\Sigma}\emb{U}_2^\mathsf{T}$, where $\emb{U}_1$ and $\emb{U}_2$ are unitary matrices and descending singular value matrix $\emb{\Sigma} = \diag(\sigma_1, \sigma_2, \cdots, \sigma_c)$.  
Then $\mathbb{E}({\small\widetilde{\emb{V}}^{(l)}}) = \emb{U}_1\emb{\Sigma}\emb{\Sigma}_{\mu}\emb{U}_2^\mathsf{T}$ where $\emb{\Sigma}_{\mu} = \diag(\mu_1, \mu_2, \cdots, \mu_c)_{0<\mu_{1 \cdots c}<1}$ is in ascending order.
\end{thm}

\begin{proof}
We focus on comparing ({\small$\widetilde{\emb{V}}^{(0)}$, ${\emb{V}}^{(0)}$}), which can be easily popularized to any convolution layer, i.e., ({\small$\widetilde{\emb{V}}^{(l)}$, ${\emb{V}}^{(l)}$}). 
Conducting SVD decomposition on {\small${\emb{V}}^{(0)}$}, we have ${\emb{V}}^{(0)} = \emb{U}_1\emb{\Sigma}\emb{U}_2^\mathsf{T}$, where {\small$\emb{U}_1$} and {\small$\emb{U}_2$} are unitary matrices of singular vectors.
Then following {\small$\emb{p}^{(k)} = {\emb{V}^{(0)}}^\mathsf{T}\emb{V}^{(0)}\emb{p}^{(k-1)}$}, we shall have {\small $\emb{p}^{(k)} = ({\emb{V}^{(0)}}^\mathsf{T}\emb{V}^{(0)})^k\emb{p}^{(0)}$}.
Replacing {\small${\emb{V}}^{(0)}$} with its SVD decomposition, we get the following equation:
\begin{sequation}
\emb{p}^{(k)} = (\emb{U}_2\emb{\Sigma}^{2k}\emb{U}_2^\mathsf{T})\emb{p}^{(0)}.
\end{sequation}%
Then we transform the projection matrix in Equation~(\ref{eq:projection}) as follows:
\begin{sequation}
\begin{aligned}
\emb{P} = \frac{\emb{p}^{(k)} {\emb{p}^{(k)}}^\mathsf{T}}{{\emb{p}^{(k)}}^\mathsf{T} \emb{p}^{(k)}} & = \frac{(\emb{U}_2\emb{\Sigma}^{2k}\emb{U}_2^\mathsf{T})\emb{p}^{(0)} {\emb{p}^{(0)}}^\mathsf{T} (\emb{U}_2\emb{\Sigma}^{2k}\emb{U}_2^\mathsf{T})}
{{\emb{p}^{(0)}}^\mathsf{T} (\emb{U}_2\emb{\Sigma}^{2k}\emb{U}_2^\mathsf{T}) (\emb{U}_2\emb{\Sigma}^{2k}\emb{U}_2^\mathsf{T})\emb{p}^{(0)} } \\
 & = \emb{U}_2\emb{\Sigma}^{2k} \frac{\emb{U}_2^\mathsf{T}\emb{p}^{(0)} {\emb{p}^{(0)}}^\mathsf{T} \emb{U}_2} 
 {{\emb{p}^{(0)}}^\mathsf{T} \emb{U}_2\emb{\Sigma}^{4k}\emb{U}_2^\mathsf{T}\emb{p}^{(0)} }\emb{\Sigma}^{2k}\emb{U}_2^\mathsf{T}.
\end{aligned}
\end{sequation}%
Let $\emb{t} = \emb{U}_2^\mathsf{T}\emb{p}^{(0)}$, we can further simplify the above equation to:
\begin{sequation}
\emb{P} = \emb{U}_2 \emb{\Sigma}^{2k} \frac{ \emb{t} \emb{t}^\mathsf{T}}{\emb{t}^\mathsf{T}\emb{\Sigma}^{4k} \emb{t} }\emb{\Sigma}^{2k} \emb{U}_2^\mathsf{T}, 
 \text{ where scalar }  \emb{t}^\mathsf{T}\emb{\Sigma}^{4k} \emb{t} = \sum_{j=1}^c t_j^2 \sigma_j^{4k}.
\end{sequation}%
Recalling that {\small $\widetilde{\emb{V}}^{(0)} = \emb{V}^{(0)}(\emb{I} - \epsilon \emb{P})$}, {\small $\mathbb{E}(\widetilde{\emb{V}}^{(0)}) = \emb{V}^{(0)} - \epsilon\cdot\mathbb{E}(\emb{V}^{(0)}\emb{P})$}.
Then we focus on the term {\small $\mathbb{E}(\emb{V}^{(0)}\emb{P})$} as follows:
\begin{sequation}
\mathbb{E}(\emb{V}^{(0)}\emb{P}) = \frac{1}{\emb{t}^\mathsf{T}\emb{\Sigma}^{4k} \emb{t}} \emb{U}_1 \emb{\Sigma}^{2k+1} \cdot \mathbb{E}(\emb{t} \emb{t}^\mathsf{T}) \cdot  \emb{\Sigma}^{2k} \emb{U}_2^\mathsf{T} 
\end{sequation}%
Since {\small$\emb{p}^{(0)}$$\sim$$\mathcal{N}(\emb{0}, \emb{I})$} and {\small $\emb{U}_2$} is a unitary matrix, thus {\small $\emb{t}$$\sim$$\mathcal{N}(\emb{0}, \emb{I})$}. 
This indices that each element of {\small $\emb{t}$}, e.g., {\small $t_j$ $\in$ $\emb{t}$}, is \textit{i.i.d.} random variable. Thus, {\small $\mathbb{E}({t}_j \cdot t_{k})=0$} for {\small$j\neq k$} and {\small $\mathbb{E}(\emb{t}\emb{t}^\mathsf{T})$} is a diagonal matrix, i.e., {\small $\mathbb{E}(\emb{t}\emb{t}^\mathsf{T})=\diag(t_1^2, t_2^2, \cdots, t_c^2)$}.
We then have:
\begin{sequation}
\mathbb{E}(\emb{V}^{(0)}\emb{P}) = \emb{U}_1 \cdot \diag \big(\frac{\sigma_1 t_1^2 \sigma_1^{4k}}{\sum_{j=1}^c t_j^2 \sigma_j^{4k}},  \cdots,  \frac{\sigma_c  t_c^2\sigma_c^{4k}}{\sum_{j=1}^c t_j^2 \sigma_j^{4k}}\big) \cdot \emb{U}_2^\mathsf{T}.
\end{sequation}%
Therefore, 
\begin{sequation}
\mathbb{E}(\widetilde{\emb{V}}^{(0)}) = \emb{U}_1 \cdot \diag \big( \sigma_1 - \epsilon \frac{\sigma_1 t_1^2 \sigma_1^{4k}}{\sum_{j=1}^c t_j^2 \sigma_j^{4k}},  \cdots,  \sigma_c - \epsilon \frac{\sigma_c  t_c^2\sigma_c^{4k}}{\sum_{j=1}^c t_j^2 \sigma_j^{4k}} \big)  \cdot \emb{U}_2^\mathsf{T}.
\end{sequation}%
Let {\small $\mu_k = 1 - \epsilon \frac{t_k^2\sigma_k^{4k}}{\sum_{j=1}^c t_j^2 \sigma_j^{4k}}$}, with {\small $\epsilon$ $\in$ $(0,1)$}, obviously, {\small $0<\mu_k<1$}. 
Furthermore, {\small $\forall k_1$ $\geq$ $k_2$}, we have:
\begin{sequation}
\label{eq:increase}
\resizebox{1\linewidth}{!}{$
\displaystyle
\mu_{k_1} - \mu_{k_2} = \epsilon \mathbb{E}(\frac{t_{k_1}^2\sigma_{k_1}^{4k}}{\sum_{j=1}^c t_j^2 \sigma_j^{4k}} - \frac{t_{k_2}^2\sigma_{k_2}^{4k}}{\sum_{j=1}^c t_j^2 \sigma_j^{4k}}) \geq \epsilon\sigma_{k_1}^{4k} \cdot \mathbb{E} (\frac{t_{k_1}^2 - t_{k_2}^2}{\sum_{j=1}^c t_j^2 \sigma_j^{4k}}) =0,
$}
\end{sequation}%
as {\small $\sigma_{k_2}^{4k} \geq \sigma_{k_1}^{4k}$}, and $t_{k_1}$ and $t_{k_2}$ are \textit{i.i.d.} random variables with same normal distribution.
Equation~(\ref{eq:increase}) proves that {\small $\mu_k$} is \textit{monotone non-decreasing} in expectation, which completes the proof.
\end{proof}